\documentclass[journal]{IEEEtran}
\usepackage{graphicx}
\usepackage{setspace}
\usepackage{amsmath,empheq}
\usepackage{amssymb}
\usepackage{amsthm}
\usepackage{caption}
\usepackage{subcaption}
\usepackage{bm} 
\usepackage{cite}
\usepackage{float}
\usepackage[usenames,dvipsnames]{pstricks}
\usepackage{epsfig}
\usepackage{pst-grad} 
\usepackage{pst-plot} 
\usepackage{tabularx}
\usepackage{amsmath}
\usepackage[T1]{fontenc}
\newcommand{\lb} {\left}
\newcommand{\rb} {\right}
\newcommand{\nn} {\nonumber}
\usepackage{xcolor}
\usepackage{lipsum}
\usepackage{comment}

\newtheorem{theorem}{Theorem}
\newtheorem{corollary}{Corollary}

\newtheorem{remark}{Remark}
\allowdisplaybreaks

\usepackage[normalem]{ulem}

\allowdisplaybreaks

\usepackage{caption}
\usepackage[utf8]{inputenc}

\usepackage{flushend}
\flushend
 
\begin{document}
\title{Opportunistic User Scheduling for Secure RIS-aided Wireless Communications}

\author{Burhan Wafai,~\IEEEmembership{Student Member,~IEEE}, Sarbani Ghose,~\IEEEmembership{Member,~IEEE}, Chinmoy Kundu,~\IEEEmembership{Member,~IEEE}, Ankit Dubey,~\IEEEmembership{Member,~IEEE}, 
and Mark F. Flanagan,~\IEEEmembership{Senior Member, IEEE}
\thanks{This work was supported in part by the Tata Consultancy Services (TCS) Foundation through its TCS Research Scholar Program,
the Science Foundation Ireland (SFI) under Grant Number 22/IRDIFA/10425,  and the Ministry of Electronics and Information Technology (MeitY), Govt. of India, through its project on capacity building for human resource development in unmanned aircraft systems (drone and related technology) with Ref. No. L14011/29/2021-HRD.
} 
\thanks{Burhan Wafai and Ankit Dubey are with the Department of Electrical Engineering, Indian Institute of Technology Jammu, Jammu  181221, India (email: burhan.wafai@iitjammu.ac.in and ankit.dubey@iitjammu.ac.in).}
\thanks{Sarbani Ghose is with DSZ Innovation Laboratories Private Limited, Kolkata, 700036, India (email: sarbani.ghose@dszlabs.com).}
\thanks{Chinmoy Kundu was with the School of Electrical and Electronic Engineering, University College Dublin, Dublin D04 V1W8, Ireland. He is now with the Wireless Communications Laboratory, Tyndall National Institute, University College Cork, Cork T12 R5CP, Ireland (email: chinmoy.kundu@tyndall.ie).}
\thanks{Mark F. Flanagan is with the School of Electrical and Electronic Engineering, University College Dublin, Dublin D04 V1W8, Ireland (email:mark.flanagan@ieee.org).}
}

\maketitle

\thispagestyle{empty}

\begin{abstract}
In this paper, we provide expressions for the secrecy outage probability (SOP) for suboptimal and optimal opportunistic scheduling schemes in a reconfigurable intelligent surface (RIS) aided {single antenna} system with multiple eavesdroppers in approximate closed form.
A suboptimal scheduling (SS) scheme is analyzed, which is used when the channel state information (CSI) of the eavesdropping links is unavailable, and the optimal scheduling (OS) scheme is also analyzed, which is used when the global CSI is available.
For each scheme, we provide a simplified expression for the SOP in the high signal-to-noise ratio (SNR) regime to demonstrate its behavior as a function of the key system parameters.
At high SNR, the SOP saturates to a constant level which decreases exponentially with the number of RIS elements in the SS scheme and with the product of the number of RIS elements and the number of users in the OS scheme.
{We also show that the derived SOP of the SS scheme can directly provide the SOP for the best antenna-user pair scheduling scheme in a multiple antenna system.}
We compare the performance of the opportunistic user scheduling schemes with that of a non-orthogonal multiple access (NOMA) based scheduling scheme which chooses a pair of users in each time slot for scheduling and we show that the opportunistic schemes outperform the NOMA-based scheme.
We also derive a closed-form expression for the SOP of a decode-and-forward (DF) relay-aided scheduling scheme in order to compare it with that of the RIS-aided system. It is found that the  RIS-aided system outperforms the relay-aided systems when the number of RIS elements is sufficiently large. An increased number of RIS elements is required to outperform the relay-aided system at higher operating frequencies.
\end{abstract}


\section{Introduction}
 A reconfigurable intelligent surface (RIS), consisting of a large number of passive reconfigurable elements (e.g., low-cost printed dipoles), is an emerging next-generation wireless communication technology that can induce a certain phase shift independently on the incident signal using smart electronic controllers.  Therefore, it can boost communication performance and coverage intelligently by shaping the propagation environment into a desired form \cite{Zhang_RIS_infty_CLT, Renzo_Tretyakov_RIS_state_of_research}. Due to the broadcast nature of wireless communications, RIS-aided systems are not immune to eavesdropping \cite{wyner_wiretap_channel, poor_Wireless_physical_layer_security, wireless_IT_security_bloch}. Towards securing RIS-aided communication, physical layer security (PLS) has gained much interest recently due to its low-complexity techniques based on physical properties of the radio channel (fading, interference, and path diversity, etc.) \cite{Renzo_Secrecy_Performance_Analysis, Ai_Ottersten_Secure_vehicular_communications,mumtaz_RIS_GC21,yadav_2021_ACTS,masoud_sop_smartgrid_TII,ghadi_2023_ris_fishersnedecor,wafai_gc,renzo_SOP_discrete_phase_TVT,Trigui_Zhu_Secrecy_Outage_Probability,Wang_Ni_Secrecy_performance_analysis, wei_shi_TWC24_sop}.

The average secrecy performance in the RIS-aided systems is analyzed using the two key performance metrics in the literature, i.e., secrecy outage probability (SOP) and ergodic secrecy rate (ESR).
The SOP provides the fraction of fading realizations for which the wireless channel can support a specific secure rate.
The SOP of an RIS-aided system with a single user and a single eavesdropper has been investigated in various studies \cite{Renzo_Secrecy_Performance_Analysis, Ai_Ottersten_Secure_vehicular_communications, mumtaz_RIS_GC21,yadav_2021_ACTS,masoud_sop_smartgrid_TII,ghadi_2023_ris_fishersnedecor,wafai_gc, renzo_SOP_discrete_phase_TVT}.
Notably, the direct link between the source and the user was taken into consideration in \cite{Renzo_Secrecy_Performance_Analysis} while the direct link was assumed to be blocked due to obstacles in \cite{Ai_Ottersten_Secure_vehicular_communications}. 
{In \cite{mumtaz_RIS_GC21}, the SOP in presence of distortion noise from residual hardware impairment was evaluated for the case where the direct link is assumed to be blocked.  The SOP, intercept probability, and probability of non-zero secrecy capacity were evaluated in \cite{yadav_2021_ACTS}, again considering the absence of the direct link.}  The ESR was evaluated along with the SOP in the presence of the direct link in \cite{masoud_sop_smartgrid_TII}. {The SOP and average secrecy capacity were evaluated over Fisher-Snedecor $\mathcal{F}$ fading channels in the absence of the direct link in \cite{ghadi_2023_ris_fishersnedecor}.}
The SOP of a jammer-supported RIS-aided system was derived in \cite{wafai_gc},  without considering the direct link. An optimal power allocation factor between the source and the jammer was obtained by minimizing the SOP.
An upper bound for the SOP was determined in \cite{renzo_SOP_discrete_phase_TVT} without the direct link under the constraint of discrete phase shifts at the RIS due to hardware limitations. 
The SOP without the direct link was obtained in the presence of multiple eavesdroppers in \cite{Trigui_Zhu_Secrecy_Outage_Probability, Wang_Ni_Secrecy_performance_analysis, wei_shi_TWC24_sop}. Specifically, non-colluding and colluding eavesdroppers with discrete phase shifts at the RIS were considered in \cite{Trigui_Zhu_Secrecy_Outage_Probability}.
An RIS-aided unmanned aerial vehicle system with eavesdroppers distributed according to the homogeneous two-dimensional Poisson point process (PPP) was investigated using stochastic geometry theory in \cite{Wang_Ni_Secrecy_performance_analysis}. Stochastic geometry theory is also applied in \cite{wei_shi_TWC24_sop} to obtain the SOP and the ESR of an RIS-aided multiple antenna communication system in the presence of random spatially distributed eavesdroppers.

The above articles consider evaluating the SOP of RIS-aided systems with a single user; however, future networks will be dense where many users will coexist along with multiple potential eavesdroppers.  With multiple users in an RIS-aided system, the SOP, ESR, and non-zero secrecy capacity were evaluated using stochastic geometry theory in \cite{jiayi_zhang_TIFS21_pls} for randomly located users in the presence of a multi-antenna eavesdropper.

In scenarios involving multiple users with uncorrelated fading exhibited by separate users, opportunistically selecting the best user that has the highest signal-to-noise ratio (SNR) among all users is a proven technique to improve the output SNR. In this user selection technique, the output SNR of the system is proportional to the individual link SNRs and the number of links \cite{jakes1994microwave}. This technique is the simplest and most inexpensive of all the diversity combining techniques. As a result, opportunistic user scheduling techniques are particularly useful in low-complexity next-generation wireless systems and have been extensively studied for improving secrecy in systems without the RIS \cite{ Bletsas_path_selection, kundu_dual_hop_regenerative, Telex_D2D_selection,   kundu_Ergodic_secrecy_rate_of_optimal,kotwalTVT}. Opportunistic user scheduling enhances secrecy by increasing the diversity in the legitimate channel.

In the context of secrecy, generally, two types of user scheduling techniques are present in the literature. When the channel state information (CSI) of the eavesdropping links is unavailable, a suboptimal scheduling (SS) scheme is implemented. When the CSI of the eavesdropping links is available, the optimal scheduling (OS) scheme is implemented.  The SS scheme schedules the user with a maximum source-to-user link rate, thereby not requiring the eavesdropping link CSI, whereas the OS scheme chooses the user with the maximum secrecy rate, which requires global CSI knowledge.

When opportunistic user scheduling is applied to RIS-aided systems, it allows one to achieve the highest achievable rate in a time slot for the scheduled user through simple RIS phase alignment toward that user without implementing computationally intensive RIS phase shift optimization methods. Opportunistic user scheduling also serves users equally on average in a rich scattering scenario.  Owing to the simplicity, opportunistic user scheduling was studied recently for an RIS-aided system in \cite{Zhong_perfana_usersel}, where the phases of the RIS elements were aligned toward the scheduled user. However, it did not study the secrecy performance of the system.  The phase alignment toward the scheduled user automatically ensures the misalignment of phases toward the eavesdroppers; hence, opportunistic user scheduling with RIS phase alignment toward the scheduled user can improve overall system security, and thus its SOP performance needs to be evaluated.

In contrast to opportunistic user scheduling, which is an orthogonal scheme, in a non-orthogonal multiple-access (NOMA) scheme aligning the phases of the RIS elements simultaneously towards each user is not possible. Furthermore,  joint optimization of the phases of the RIS elements and the power allocation for multiple users is also computationally intensive in the NOMA-based scheduling. 
The secrecy of an RIS-aided system  with NOMA scheduling has been studied in \cite{NOMA_secrecy_RIS_GC_21,yang_RIS_NOMA_SOP,ghadi2024physicallayersecurityperformance,yang_sop_noma_TVT23, pei_sop_noma_TVT23}. 
In a distributed RIS-aided NOMA network consisting of a source, multiple users, and an eavesdropper, the minimum secrecy rate among all of the users was maximized utilizing the SOP constraint derived for individual users in \cite{NOMA_secrecy_RIS_GC_21}. 
{The SOP of a multiple-RIS-aided NOMA system was obtained in \cite{yang_RIS_NOMA_SOP} while the SOP along with the average secrecy capacity and secrecy energy efficiency were evaluated in \cite{ghadi2024physicallayersecurityperformance}.
An RIS-aided covert communication system with a source, two users, and an eavesdropper was considered in \cite{yang_sop_noma_TVT23}, where the source employs NOMA and rate-splitting to communicate with the users. The approximate SOP of the system was derived with fixed power allocation.}
The secrecy of an RIS-aided NOMA network consisting of multiple users and a single eavesdropper was considered in \cite{pei_sop_noma_TVT23}. An approximate SOP  for individual users was evaluated with an on-off control RIS phase shift design to show the comparison between the NOMA  and orthogonal multiple access. It was concluded that the RIS-aided NOMA network has a better performance.
However, the above literature \cite{yang_sop_noma_TVT23,ghadi2024physicallayersecurityperformance,NOMA_secrecy_RIS_GC_21,yang_RIS_NOMA_SOP,pei_sop_noma_TVT23} did not investigate NOMA-based scheduling for secure communication.
Moreover, a comprehensive comparison between opportunistic user scheduling and NOMA-based scheduling is notably missing from the existing literature.

In the context of secure RIS-aided communication, it is important to determine when the RIS-aided system outperforms a relay-aided system to make deployment decisions between RIS and relay-aided systems. A comprehensive comparison between the RIS and relay-aided systems was presented in \cite{emil_RIS_relay_comparison, Ntontin_relay_ris, Renzo_ris_relay_ojcoms,   Alouini_ris_relay_ojcoms,madrid_DC_relay_v_ris}. 
A common observation in these articles is that the RIS requires a specific number of elements to outperform the corresponding relay-aided system except in \cite{madrid_DC_relay_v_ris}, where a mobile relay with multiple antennas generally outperformed the corresponding RIS-aided system due to the higher degrees of freedom in the mobility of the relay against the static RIS. 
However, whether similar observations hold in secure communication remains unknown, as a comparison between the RIS-aided system and the relay-aided system in this context is not available. 
In RIS-aided secure communication, a comparison of the SOP for an RIS and the relay-aided system was presented for a simple single-user single-eavesdropper case in \cite{Renzo_Secrecy_Performance_Analysis}. Though the RIS-aided system outperformed the relay-aided system for a fixed number of RIS elements, it was not shown how many RIS elements were required to outperform the relay-aided system.

Although opportunistic user scheduling is the simplest and most inexpensive multi-user diversity combining technique and can be easily integrated into RIS-aided systems through simple RIS phase alignment towards the best user, it has not yet received sufficient attention in RIS-aided systems for secrecy enhancement. Only recently, \cite{wcnc_2024} considered opportunistic user scheduling using the SS scheme in RIS-aided systems. The SOP was obtained in the presence of a single eavesdropper and compared with that of an RIS-aided NOMA-based scheduling scheme.
In this work, we generalize the analysis of \cite{wcnc_2024} to the case of SS \textit{and} OS schemes in the presence of \textit{multiple eavesdroppers}. We also provide more comprehensive results compared with \cite{wcnc_2024}, including a comparison with the NOMA-based system and the DF relay-aided system, which are necessary for making deployment decisions.
It is also worth mentioning that all of the aforementioned articles utilized a simple distance-dependent path loss model. However, a realistic path loss model should be incorporated into the RIS-aided system analysis to quantify the SOP performance accurately. 

Motivated by the above discussion, we consider an RIS-aided system consisting of a single source and multiple users in the presence of multiple eavesdroppers where each node is equipped with a single antenna. We analyze two opportunistic user scheduling schemes where the phases of the RIS are aligned toward the scheduled user to improve the secrecy of the RIS-aided system.
In this context, our contributions are summarized as follows:
\begin{itemize}
    \item A generalized SOP analysis for the SS and OS opportunistic user scheduling schemes is presented in approximate closed form that incorporates multiple users and multiple eavesdroppers {in the single antenna system. The analysis can also directly provide the SOP of the best antenna-user pair selection scheme in the multiple antenna system.}
   A realistic path loss model is adopted that considers frequency, distances, and angles of incidence and reflection at the RIS  following, \cite{nemanja_TWC} and \cite{Fadil_pathloss_ris_green_theorem}. 
    
   \item Simplified closed-form high-SNR expressions are provided for each scheduling scheme to demonstrate how the secrecy performance depends on the system parameters. This offers insights as well as guidelines for system design to determine the number of RIS elements required and where to place an RIS  to achieve a specific level of secrecy in the system. 
      
   \item It is demonstrated that in the high-SNR regime, the SOP saturates to a constant level. The saturation level decreases exponentially with the number of RIS elements and with the product of the number of RIS elements and users in the SS scheme and in the OS scheme, respectively. 

   \item The RIS-aided scheduling schemes are compared with a specific NOMA-based scheduling scheme, wherein a pair of NOMA users is scheduled in each time slot. The results show that the secrecy performance of the opportunistic scheduling schemes outperforms the NOMA-based scheduling scheme.
   
  \item The SOP performance of the RIS-aided scheduling schemes is compared with that of a relay-aided scheduling scheme in two scenarios: i)  when a direct link between the source and the users is unavailable and ii) when it is available. We show that the RIS-aided system outperforms the corresponding relay-aided system only when the RIS has a critical number of elements. This number also depends on the frequency of operation.
 \end{itemize}

\textit{Notation:}  
$\mathbb{P[\cdot]}$ denotes the probability of an event, $\mathbb{E}\left[\cdot \right]$ denotes the expectation operator and $\mathbb{V}\left[\cdot \right]$ denotes the variance operator.
The probability density function (PDF) and the cumulative distribution function (CDF) of a random variable $X$ are denoted by $f_{X}(\cdot)$ and $F_{X}(\cdot)$, respectively.  
$X\sim \mathcal{CN}(\mu,\sigma^2)$ denotes the distribution of a complex Gaussian random variable $X$ with mean $\mu$ and variance $\sigma^2$, and $\sim$ stands for ``distributed as''. $\max\{\cdot\}$ and
$\min\{\cdot\}$ denote the maximum and minimum of its arguments,
respectively.



\section{System Model}
\label{sec_System_Model}

\begin{figure}
\centering
\includegraphics[width=0.485\textwidth]{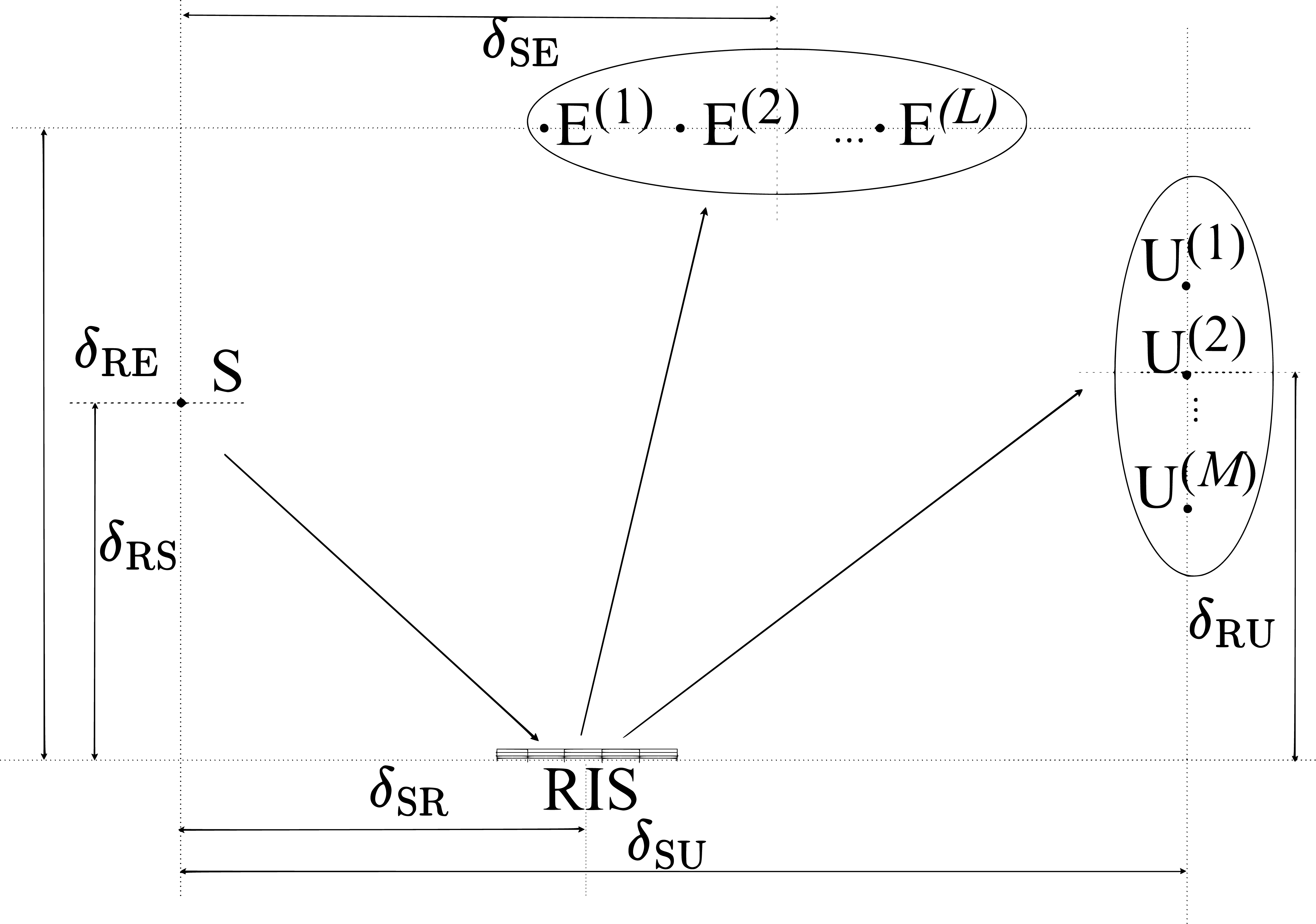}
\caption{An RIS-aided system with multiple users and multiple eavesdroppers.}
\label{FIG_1_systems_ariel}
\end{figure} 

We consider a wireless communication system where an RIS  with $N$ reflecting elements is providing an indirect communication link between a single-antenna source S and $M$ single-antenna users $\text{U}^{(m)}$ for $m\in\mathcal{M}=\{1,2,\ldots, M\}$, in presence of $L$  passive single-antenna eavesdroppers  $\text{E}^{(l)}$ for $l\in\mathcal{L}=\{1,2,\ldots, L\}$. 
{We assume that the direct links from the source to users and from the source to eavesdroppers are absent due to obstacles. An RIS is particularly useful for system performance improvement when the direct link is unavailable, and this use case is assumed commonly in the RIS literature (c.f., \cite{mumtaz_RIS_GC21,yadav_2021_ACTS,wafai_gc,Renzo_Secrecy_Performance_Analysis,masoud_sop_smartgrid_TII,Wang_Ni_Secrecy_performance_analysis,ghadi_2023_ris_fishersnedecor,Ai_Ottersten_Secure_vehicular_communications,renzo_SOP_discrete_phase_TVT, Trigui_Zhu_Secrecy_Outage_Probability,wei_shi_TWC24_sop,Zhang_RIS_infty_CLT,nemanja_TWC}).}
We assume that the nodes S, $\text{U}^{(m)}$ for $m\in\mathcal{M}$, $\text{E}^{(l)}$ for $l\in\mathcal{L}$, and the RIS are placed in different vertical planes as depicted in Fig. \ref{FIG_1_systems_ariel}. The plane containing the users is parallel to the plane containing S. The plane containing the RIS is perpendicular to the planes containing S and the users; and parallel to the plane containing the eavesdroppers.



The distance between the plane containing S and the center of the RIS is $\delta_{\text{SR}}$, the distance between the planes containing S and $\text{U}^{(m)}$ for any $m$ is $\delta_{\text{SU}}$, and the distance between the plane containing S and $\text{E}^{(l)}$ for any $l$ is $\delta_{\text{SE}}$. Similarly, we denote the distance between the plane containing the RIS and $\text{S}$ as $\delta_{\text{RS}}$, the distance between the plane containing the RIS and $\text{U}^{(m)}$ for any $m$ as $\delta_{\text{RU}}$, and the distance between the planes containing the RIS and $\text{E}^{(l)}$ for any $l$ as $\delta_{\text{RE}}$. The distances S-RIS, $\text{RIS-U}^{(m)}$ for any $m$,  and $\text{RIS-E}^{(l)}$ for any $l$, are denoted as $d_{\text{SR}}$, $d_{\text{RU}}$, and $d_{\text{RE}}$, respectively, and can be easily obtained from the geometry of the system.
Our analysis can be extended to the case of more general user positions where each user distance from the RIS is different. In such a scenario,  a fairness-aware scheduling method, such as proportional fair (PF) scheduling is commonly employed  \cite{PF_scheduling}. However, when equal user distances are considered from the RIS, the opportunistic user scheduling scheme proposed in this paper performs as effectively as the PF scheduling scheme. This special case we consider allows us to derive closed-from analytical equations from which we can derive useful insights. Indeed, the literature considering multi-user RIS-aided systems frequently adopts this assumption \cite{jiayi_zhang_TIFS21_pls,Zhong_perfana_usersel,wcnc_2024}.
\color{black}

We assume that node distances from the RIS are such that the far-field model is valid for signal transmission via the RIS \cite{nemanja_TWC,Fadil_pathloss_ris_green_theorem,Ellingson_RIS_pathloss}. All the elements of the RIS have the size of $\frac{\lambda}{2}\times\frac{\lambda}{2}$, where $\lambda$ is the wavelength of operation. The distance between the centers of adjacent RIS elements in both dimensions is $\frac{\lambda}{2}$.
{Typically, the non-negligible channel correlation among RIS reflecting elements is present when the electrical size of the reflecting elements is between $\frac{\lambda}{8}$ and $\frac{\lambda}{4}$ \cite{jiayi_zhang_TIFS21_pls}. If the electrical size of the RIS reflecting elements is larger than $\frac{\lambda}{4}$, channel correlation among RIS reflecting elements becomes weak. Thus, in this work we ignore the channel correlation among the RIS elements.}

The source signal reflected from the RIS is received by $M$ users as well as $L$  eavesdroppers. The phase contribution introduced by the $n$-th reflective element of the RIS is $\theta^{(n)}$, where $n\in\{1,2,\ldots,N\}$, and the corresponding reflective coefficient is  $\eta$ for each $n$. The free space path loss for the indirect link from S to $\text{U}^{(m)}$ for any $m$
via the RIS is written as  \cite{Fadil_pathloss_ris_green_theorem, nemanja_TWC}
\begin{align}\label{path_loss_angle}
\zeta_{\text{SU}} = \frac{\lambda^4}{256\pi^2}\frac{\big(\cos{\gamma^{(\text{inc})}}+\cos{\gamma^{(\text{ref})}_{\text{U}}}\big)^2}{d_{\text{SR}}^2d_{\text{RU}}^2},
\end{align}
where $\gamma^{(\text{inc})}$ represents the smallest angle formed by the vector normal to the RIS and the wavefront of the electromagnetic wave that originates from S,
and $\gamma^{(\text{ref})}_{\text{U}}$ represents the smallest angle formed by the vector normal to the RIS and the wavefront of the electromagnetic wave that is reflected from the RIS and observed at users.
Here we note that $d_{\text{SR}}=\sqrt{\delta_{\text{SR}}^2+\delta_{\text{RS}}^2}$ and
$d_{\text{RU}}=\sqrt{(\delta_{\text{SU}}-\delta_{\text{SR}})^2+\delta_{\text{RU}}^2}$. Similarly, the free space path loss for the indirect link between S and  $\text{E}^{(l)}$ for any $l$ can also be obtained.

The small scale fading between S and the $n$-th RIS element,  the $n$-th RIS element and $\text{U}^{(m)}$, and  the $n$-th RIS element and $\text{E}^{(l)}$ are denoted as $h_{\text{SR}}^{(n)}$,  $h_{\text{RU}}^{(n,m)}$,  and $h_{\text{RE}}^{(n,l)}$, respectively.   We assume that $h_{\text{SR}}^{(n)}$,  $h_{\text{RU}}^{(n,m)}$,  and $h_{\text{RE}}^{(n,l)}$ for any $n$, $m$, and $l$ are independent complex Gaussian random variables (RVs) with zero mean and unit variance following related channel models in \cite{Renzo_Secrecy_Performance_Analysis,wafai_gc, Trigui_Zhu_Secrecy_Outage_Probability,masoud_sop_smartgrid_TII,Wang_Ni_Secrecy_performance_analysis,Ai_Ottersten_Secure_vehicular_communications}. \footnote{{
The SOP analysis methodology used in this paper can be extended to other possible channel models, such as Rician fading.}}
The received signals at $\text{U}^{(m)}$  and $\text{E}^{(l)}$ 
via the $N$ RIS reflecting elements are written as
\begin{align}
\label{eq_ReceivedSignal_D}
y_{\text{U}}^{(m)}&=\sqrt{P\zeta_{\text{SU}}}\sum_{n=1}^N  h_{\text{RU}}^{(n,m)} \eta\exp(j\theta^{(n)}) h_{\text{SR}}^{(n)}s+\epsilon_{\text{U}}^{(m)},\\
\label{eq_ReceivedSignal_E}
y_{\text{E}}^{(l)}&=\sqrt{P\zeta_{\text{SE}}}\sum_{n=1}^N  h_{\text{RE}}^{(n,l)} \eta\exp(j\theta^{(n)}) h_{\text{SR}}^{(n)}s+\epsilon_{\text{E}}^{(l)},
\end{align}
respectively, where $s$ is the unit-energy symbol for transmission, $P$ is the transmitted power per symbol,  $\epsilon_\text{U}^{(m)}\sim \mathcal{CN}(0,N_0)$ and $\epsilon_\text{E}^{(l)}\sim \mathcal{CN}(0,N_0)$  are the complex additive white Gaussian noise (AWGN) with mean zero and variance $N_0$ at the node $\text{U}^{(m)}$ and $\text{E}^{(l)}$, respectively, for each $m\in\mathcal{M}$ and $l\in\mathcal{L}$.
The received SNR at $\text{U}^{(m)}$ is written as
\begin{align}
\label{eq_SNR_general}
\Gamma_{\text{U}}^{(m)}= \eta\bar{\Gamma}_{\text{U}}\lb\lvert\sum_{n=1}^N|h_{\text{RU}}^{(n,m)}||h_{\text{SR}}^{(n)}|\exp(\psi^{(n,m)}+\phi^{(n)}+\theta^{(n)})\rb\rvert^2,
\end{align}
where $|h_{\text{RU}}^{(n,m)}|$ and $|h_{\text{SR}}^{(n)}|$ are the magnitudes and, $\psi^{(n,m)}$ and $\phi^{(n)}$ are the phases of the channels $h_{\text{RU}}^{(n,m)}$ and $h_{\text{SR}}^{(n)}$, respectively, and $\bar{\Gamma}_{\text{U}}=\frac{P\zeta_{\text{SU}}}{N_0}$. 


To improve the system's secrecy, one out of $M$ users is selected opportunistically for scheduling. 
In the case where the $m$-th user is scheduled, the phase shifts of the RIS elements are configured in such a way that the received SNR of the $m$-th user is maximized in (\ref{eq_SNR_general}). Hence, the phase shifts of the RIS elements are configured as $\theta^{(n)}=-(\phi^{(n)}+\psi^{(n,m)})$ for each $n$ to align the phases of 
S-RIS and RIS-$\textrm{U}^{(m)}$ channels. 
{While this RIS phase configuration may not be optimal from the SOP perspective, its simplicity allows us to derive closed-form expressions for the SOP of user scheduling schemes, which provides useful insights into the secrecy performance.}
In this case, the received SNR for the $m$-th user in (\ref{eq_SNR_general}) is written as $\Gamma_{\text{U}}^{(m)}=\bar{\Gamma}_{\text{U}} \lvert{Y_{\text{U}}^{(m)}}\rvert^2$
where $Y_{\text{U}}^{(m)}=\eta\sum^N_{n=1} \lvert h_{\text{RU}}^{(n,m)}\rvert \lvert h_{\text{SR}}^{(n)}\rvert$.
When $N$ is sufficiently large,  the distribution of $Y_{\text{U}}^{(m)}$ can be tightly approximated by a Gaussian distribution due to the central limit theorem (CLT) and hence, the distribution of $\Gamma_{\text{U}}^{(m)}$ can be approximated with the help of a
non-central Chi-squared distribution with a single degree of freedom \cite{book_papoulis_pillai, Kudathanthirige_Amarasuriya_IRS_rayleigh}.
 The approximate CDF of $\Gamma_{\text{U}}^{(m)}$ is written from \cite{Kudathanthirige_Amarasuriya_IRS_rayleigh}
as
\begin{align}
\label{eq_cdf_gamma_m_inf}
F_{\Gamma_{\text{U}}^{(m)}}(x)
&=1-\xi Q\Big(\frac{\sqrt{{x}/{\bar{\Gamma}_{\text{U}}}}-\mu_{\textrm{U}}}{\sigma_{\text{U}}}\Big) ~~\textrm{for~} x\ge0,
\end{align}
\sloppy where 
 $\mu_{\textrm{U}}
= {\eta N\pi}/{4}$  
is the mean of $Y_{\text{U}}^{(m)}$,  
 $\sigma_{\text{U}}
=\eta\sqrt{N({16-\pi^2})/{16}}$ is the standard deviation of $Y_{\text{U}}^{(m)}$, 
$\xi={1}/{{Q}\big(-{\mu_{\textrm{U}}/\sigma_{\text{U}}}\big)}$   
is a normalization coefficient to satisfy the PDF constraint
$\int_{0}^{\infty}f_{\Gamma_{\text{U}}^{(m)}}(x)dx=1$, and ${Q}(\cdot)$ is the Gaussian $Q$-function. \footnote{{
In case of Rician fading channel, the closed-form SOP analysis will remain the same as proposed in the paper, with only mean $\mu_{\textrm{U}}$ and standard deviation $\sigma_{\text{U}}$ of $Y_{\text{U}}^{(m)}$, which can be derived from \cite{Rician_ris}, changing to reflect the Rician fading channel environment. This is because our SOP analysis uses the CLT to approximate the distribution of $\Gamma_{\text{U}}^{(m)}$. Similarly, the SOP analysis can also be extended to other fading environments.}}

As RIS phase shifts are
configured to align the phases of 
S-RIS and RIS-$\textrm{U}^{(m)}$ channels to maximize the SNR of $\text{U}^{(m)}$, the phases of S-RIS and RIS-$\textrm{E}^{(l)}$  channels will not be aligned for any eavesdroppers.
Due to this misalignment, the received SNR at the $l$-th eavesdropper in the case where the $m$-th user is scheduled is written following (\ref{eq_SNR_general}) as
\begin{align}
\label{eq_SNR_E_no_alighnment}
\Gamma_{\text{E}}^{(m,l)}&= \bar{\Gamma}_{\text{E}}\lb\lvert\sum_{n=1}^N\eta|h_{\text{RE}}^{(n,l)}||h_{\text{SR}}^{(n)}|\exp(\omega^{(n,l)}-\psi^{(n,m)})\rb\rvert^2\nn\\
&=\bar{\Gamma}_{\text{E}}\lvert Y_{\text{E}}^{(m,l)}\rvert^2,
\end{align}
where  $|h_{\text{RE}}^{(n,l)}|$ and $\omega^{(n,l)}$ are the magnitude and phase of $h_{\text{RE}}^{(n,l)}$ and $\bar{\Gamma}_{\text{E}}=\frac{P\zeta_{\text{SE}}}{N_0}$.
When $N$ is sufficiently large, the distribution of  $Y_{\text{E}}^{(m,l)}$ in the above equation approximates a complex Gaussian distribution due to the CLT, and hence, $\Gamma_{\text{E}}^{(m,l)}$ is exponentially distributed with PDF $f_{\Gamma_{\text{E}}^{(m,l)}}(x)=\frac{1}{\lambda_{\text{E}}}\exp({-\frac{x}{\lambda_{\text{E}}}})$
and average SNR $\lambda_{\text{E}}=\eta^2N\bar{\Gamma}_{\text{E}}$\cite{Renzo_Secrecy_Performance_Analysis}.

The secrecy is measured against the eavesdropper with maximum eavesdropping SNR. When the $m$-th user is scheduled, the maximum eavesdropping SNR is denoted as $\Gamma_{\text{E}}^{(m)}=\max_{l\in\mathcal{L}}\{\Gamma_{\text{E}}^{(m,l)}\}$.
The PDF of $\Gamma_{\text{E}}^{(m)}$ is obtained  by finding the CDF of $\Gamma_{\text{E}}^{(m)}$. The CDF of $\Gamma_{\text{E}}^{(m)}$ is derived as
\begin{align}
\label{eq_cdf_Z}
F_{\Gamma_{\text{E}}^{(m)}}(x)&
=\mathbb{P}\Big[\max_{l\in\mathcal{L}}\{\Gamma_{\text{E}}^{(m,l)}\} \le x\Big]\nn\\
&=1-\sum^L_{l=1}(-1)^{l+1}\binom{L}{l} \exp\Big(-\frac{ x}{\lambda_{\text{E}}^{(l)}}\Big), 
\end{align}
where $\lambda_{\text{E}}^{(l)}=\lambda_{\text{E}}/l$. 
The corresponding PDF is then obtained as 
\begin{align}
\label{eq_pdf_E}
f_{\Gamma_{\text{E}}^{(m)}}(x)&=
\sum^L_{l=1}(-1)^{l+1}\binom{L}{l} \hat{f}_{\Gamma_{\text{E}}^{(l)}}(x),
\end{align}
\sloppy where  $\hat{f}_{\Gamma_{\text{E}}^{(l)}}(x)=\frac{1}{\lambda_{\text{E}}^{(l)}} 
\exp(-\frac{ x}{\lambda_{\text{E}}^{(l)}})$. The secrecy rate achievable by $\text{U}^{(m)}$ is $C_S^{(m)}=\max\{\log_2 (\frac{1+\Gamma_{\text{U}}^{(m)}}{1+\Gamma_{\text{E}}^{(m)}} ),0\}$.

We propose two criteria for opportunistic user scheduling depending on the availability of the eavesdropping CSI. If the CSI of the eavesdroppers is not available, the user with the maximum achievable rate is scheduled such that the SNR at the scheduled user $\textrm{U}^{(m^*)}$ is written as  $\Gamma_{\text{U}}^{(m^*)}=\max_{ m\in\mathcal{M}}\{\Gamma_{\text{U}}^{(m)}\}$. This scheme is defined as the SS scheme. Otherwise,  the user $\textrm{U}^{(m^*)}$ that achieves the maximum secrecy rate $\max_{m\in\mathcal{M}}\{C_S^{(m)}\}$ among all of the users is scheduled. This scheme is defined as the OS scheme.
 Following \eqref{eq_cdf_gamma_m_inf}, the average received SNR of the individual users  $\mathbb{E}[\Gamma_{\text{U}}^{(m)}]=\frac{\eta \bar{\Gamma}  N}{4}\lb(\pi+\eta (16-\pi)^2/4\rb)$ is large when the number of RIS elements is large. As the average received SNR of the selected user is proportional to the average received SNR of the individual users for selection diversity \cite{jakes1994microwave}, opportunistic user scheduling will ensure system performance improvement with the increased number of RIS elements.



We evaluate the SOP for the RIS-aided user scheduling criteria. The SOP of the system is defined as the probability that the achievable secrecy rate of the system falls below a certain threshold. The SOP of the system when $\text{U}^{(m^*)}$ is scheduled is evaluated as \cite{basic_sec_paper} 
\begin{align}
\label{eq_exact_CLT_sop_no_sel_mult_eve}
\mathcal{P}_{\textrm{out}}
&=\mathbb{P}\lb[\frac{1+\Gamma_{\text{U}}^{(m^*)}}{1+\Gamma_{\text{E}}^{(m^*)}}<2^{R_{\textrm{th}}}\rb] \nn\\
&=\int_0^{\infty}F_{\Gamma_{\text{U}}^{(m^*)}}\lb(\rho x+\rho-1\rb)f_{\Gamma_{\text{E}}^{(m^*)}}(x)dx,
\end{align}
where $R_{\textrm{th}}$ is the threshold secrecy rate,  $\rho=2^{R_{\textrm{th}}}$, $F_{\Gamma_{\text{U}}^{(m^*)}}(x)$ is the CDF of the SNR of the scheduled user, and $f_{\Gamma_{\text{E}}^{(m^*)}}(x)$ is the PDF of the corresponding eavesdropping SNR.


{
We assume perfect CSI of source-to-RIS and RIS-to-user channels is available at the RIS.
This assumption makes the SNR distribution at the scheduled user mathematically tractable as it allows perfect phase alignment for the source-to-RIS and RIS-to-user channels in (\ref{eq_SNR_general}). The corresponding eavesdropping SNR and the SOP analysis also become mathematically tractable. The results obtained with this assumption can act as a theoretical benchmark for the secrecy performance of a practical system.  
Similar assumptions of perfect CSI availability are also commonly considered in the existing literature\cite{Ai_Ottersten_Secure_vehicular_communications, Renzo_Secrecy_Performance_Analysis,mumtaz_RIS_GC21,yadav_2021_ACTS,jiayi_zhang_TIFS21_pls,wafai_gc,wcnc_2024,Kudathanthirige_Amarasuriya_IRS_rayleigh,ghadi_2023_ris_fishersnedecor,masoud_sop_smartgrid_TII,renzo_SOP_discrete_phase_TVT,wei_shi_TWC24_sop,Trigui_Zhu_Secrecy_Outage_Probability,Yu_Schober_Robust_and_Secur_Wireless,Wang_Ni_Secrecy_performance_analysis,Fadil_pathloss_ris_green_theorem,Ellingson_RIS_pathloss,nemanja_TWC}. }
We also note that to obtain the CDF of $\Gamma_{\text{E}}^{(m)}$ in (\ref{eq_cdf_Z}) and the SOP in (\ref{eq_exact_CLT_sop_no_sel_mult_eve}), we have assumed independence of $\Gamma_{\text{E}}^{(m,l)}$ for each $l$, independence of individual $\Gamma_{\text{U}}^{(m)}$ and $\Gamma_{\text{E}}^{(m)}$  for any $m$, and independence between $\Gamma_{\text{U}}^{(m)}$ and $\Gamma_{\text{E}}^{(m)}$ for each $m$, even though these are correlated due to the common S-R link. 
Therefore, we assume the independence of SNRs to deduce important system behaviors with the system parameters without unnecessarily complicating the analysis.  The  RIS literature frequently adopts this assumption \cite{Ai_Ottersten_Secure_vehicular_communications, Renzo_Secrecy_Performance_Analysis,wafai_gc,wcnc_2024,Kudathanthirige_Amarasuriya_IRS_rayleigh,masoud_sop_smartgrid_TII,renzo_SOP_discrete_phase_TVT,Trigui_Zhu_Secrecy_Outage_Probability,Yu_Schober_Robust_and_Secur_Wireless,Wang_Ni_Secrecy_performance_analysis,nemanja_TWC}.

\section{User Scheduling without Eavesdroppers' CSI}\label{subsec_clt_sop_subopt_usrsel_mult_eve}
This section considers user scheduling when the eavesdroppers' CSI is unavailable.  In this scenario, we implement the SS scheme where the user is scheduled for which the 
instantaneous source-to-user SNR is maximum (without considering the SNR of the corresponding eavesdropping link). The SNR at the scheduled user $\textrm{U}^{(m^*)}$ is written as $\Gamma_{\text{U}}^{(m^*)}=\max_{ m\in\mathcal{M}}\{\Gamma_{\text{U}}^{(m)}\}$.  
The advantage of the SS scheme is that it does not require any eavesdropper channel CSI. The SOP  is obtained using the following theorem.

\begin{theorem}\label{theorem2}
An approximate  closed-form expression for the SOP of the SS scheme is given by 
\begin{subequations}
\label{eq_sop_suboptimal_final}
\begin{empheq}[left={\mathcal{P}_{\textrm{out}}=\empheqlbrace}]{align}
&1-\sum^M_{m=1}\sum^L_{l=1}(-1)^{m+l} \binom{M}{m}\times \nn\\
&\binom{L}{l}\xi^m J_{+}^{(m,l)} \hspace{2.2cm}\textrm{if~}\bar{\Gamma}_{\text{U}}\le\frac{\rho-1}{\mu_{\textrm{U}}^2} \label{eq_sop_suboptimal_final1} \\
&1-\sum^M_{m=1}\sum^L_{l=1}(-1)^{m+l}\binom{M}{m}\times \nn\\
&\binom{L}{l} \xi^m\Big[I_{-}^{(m,l)}
+I_+^{(m,l)}\Big] \hspace{0.45cm}~\textrm{if~}\bar{\Gamma}_{\text{U}}>\frac{\rho-1}{\mu_{\textrm{U}}^2},
\label{eq_sop_suboptimal_final2}
\end{empheq}    
\end{subequations} 
where
\begin{align}
\label{eq_SS_J1}
&J_{+}^{(m,l)}=\sum_{\boldsymbol{k} \in \mathcal{S}_m}\binom{m}{\boldsymbol{k}}
\frac{ w_1^{k_1}w_2^{k_2}w_3^{k_3}}{2^{k_1+k_2+k_3-1}}
{\mathcal{J}_{+}^{(\boldsymbol{k},l)}((\sigma^{(\boldsymbol{k})}_{\text{U}})^2)},
\\
\label{eq_SS_I1}
&I_{+}^{(m,l)}=\sum_{\boldsymbol{k} \in \mathcal{S}_m}\binom{m}{\boldsymbol{k}}
\frac{ w_1^{k_1}w_2^{k_2}w_3^{k_3}}{2^{k_1+k_2+k_3-1}}{\mathcal{I}_{+}^{(\boldsymbol{k},l)}((\sigma^{(\boldsymbol{k})}_{\text{U}})^2)},\\
\label{eq_ss_I_minus_new_theorem}
&I_{-}^{(m,l)}=1-e^{\frac{-\mu^2_{\text{U}}\bar{\Gamma}_{\text{U}}+\lb(\rho-1\rb)}{\rho\lambda_{\text{E}}^{(l)}}}+\sum_{j=1}^{m}(-1)^j\binom{m}{j}
\big(J_+^{(j,l)}-I_+^{(j,l)}\big) ,\\
\label{eq_ss_j_new_theorem}
&\mathcal{J}_{+}^{(\boldsymbol{k},l)}((\sigma^{(\boldsymbol{k})}_{\text{U}})^2) =\frac{ \bar{\Gamma}_{\text{U}}}{2\rho\lambda_{\text{E}}^{(l)}\Upsilon^{(\boldsymbol{k},l)}}
 \Bigg[\exp\lb(-\frac{1}{2(\sigma^{(\boldsymbol{k})}_{\text{U}})^2}\rb.\nn\\
 &\times\lb.\lb(\sqrt{\frac{\rho-1}{\bar{\Gamma}_{\text{U}}}}-\mu_{\textrm{U}}\rb)^2\rb)+\frac{2\mu_{\textrm{U}}\sqrt{\pi}}{2(\sigma^{(\boldsymbol{k})}_{\text{U}})^2\sqrt{\Upsilon^{(\boldsymbol{k},l)}}} \nn \\  
&\times \exp\lb(\frac{\rho-1}{\rho\lambda_{\text{E}}^{(l)}}-\frac{\mu^2_{\text{U}} \bar{\Gamma}_{\text{U}}}
{2(\sigma^{(\boldsymbol{k})}_{\text{U}})^2\rho\lambda_{\text{E}}^{(l)}\Upsilon^{(\boldsymbol{k},l)}}\rb)  \nn \\  
&\times Q\lb(\sqrt{2\Upsilon^{(\boldsymbol{k},l)}}
\lb(\sqrt{\frac{\rho-1}{\bar{\Gamma}_{\text{U}}}}
-\frac{\mu_{\textrm{U}}}{2(\sigma^{(\boldsymbol{k})}_{\text{U}})^2\Upsilon^{(\boldsymbol{k},l)}}\rb)\rb)
\Bigg],\\
\label{eq_SS_I_new}
&\mathcal{I}_{+}^{(\boldsymbol{k},l)}((\sigma^{(\boldsymbol{k})}_{\text{U}})^2)= \frac{ \bar{\Gamma}_{\text{U}}}{2\rho\lambda_{\text{E}}^{(l)}\Upsilon^{(\boldsymbol{k},l)}}
\lb[\exp\lb(-\frac{\mu^2_{\text{U}} \bar{\Gamma}_{\text{U}}-\lb(\rho-1\rb)}{\rho\lambda_{\text{E}}^{(l)}}\rb)
\rb.\nn\\&
\lb.+\frac{2\mu_{\textrm{U}}\sqrt{\pi}}{2(\sigma^{(\boldsymbol{k})}_{\text{U}})^2\sqrt{\Upsilon^{(\boldsymbol{k},l)}}}
\exp\lb(\frac{ \rho-1}{\rho\lambda_{\text{E}}^{(l)}}-\frac{\mu^2_{\text{U}} \bar{\Gamma}_{\text{U}}}
{2(\sigma^{(\boldsymbol{k})}_{\text{U}})^2\rho\lambda_{\text{E}}^{(l)}\Upsilon^{(\boldsymbol{k},l)}}\rb)\nn\rb.\\
&\times\lb.Q\lb(\frac{\sqrt{2}\mu_{\textrm{U}} \bar{\Gamma}_{\text{U}}}{\rho\lambda_{\text{E}}^{(l)}\sqrt{\Upsilon^{(\boldsymbol{k},l)}}
}\rb)
\rb],
\end{align}
$\mathcal{S}_m$ is the set of integer vectors $\boldsymbol{k}=[k_1, k_2, k_3]$   such that $k_{i}\in\{0 ,\ldots, m\}$ for each $i\in\{1 ,2, 3\}$ and $\sum_{i=1}^3k_i=m$, 
$\binom{m}{\boldsymbol{k}}=  \frac{m!}{ k_1!k_2!k_{3}!}$, {$\sigma^{(\boldsymbol{k})}_{\text{U}}= \frac{\sigma_{\text{U}}}{\sqrt{\sum_{i=1}^3k_ip_i}}$},
and
$\Upsilon^{(\boldsymbol{k},l)}= \frac{1}{2(\sigma^{(\boldsymbol{k})}_{\text{U}})^2}+\frac{\bar{\Gamma}_{\text{U}}}{\rho\lambda_{\text{E}}^{(l)}}$.
\end{theorem}

\begin{proof}
The evaluation of the SOP in (\ref{eq_exact_CLT_sop_no_sel_mult_eve}) requires the distribution of the scheduled user channel SNR $\Gamma_{\text{U}}^{(m^*)}$ 
and the corresponding eavesdropping channel SNR.
The CDF $F_{\Gamma_{\text{U}}^{(m^*)}}(x)$  is obtained  with the help of (\ref{eq_cdf_gamma_m_inf}) as
\begin{align}
\label{eq_SS_cdf}
&F_{\Gamma_{\text{U}}^{(m^*)}}(x)
=\mathbb{P}\Big[\max_{ m\in\mathcal{M}}\{\Gamma_{\text{U}}^{(m)}\}\le x\Big]\nn\\
&
=1-\sum^M_{m=1}(-1)^{m+1}\binom{M}{m}\xi^m \Bigg[Q\Big(\frac{\sqrt{{x}/{\bar{\Gamma}_{\text{U}}}}-\mu_{\textrm{U}}}{\sigma_{\text{U}}}\Big)\Bigg]^m.
\end{align}
As the SNR for each eavesdropper is identical for any scheduled user due to the independent and identically distributed RIS-$\textrm{U}^{(m)}$  channels for any $m\in \mathcal{M}$ and RIS-$\textrm{E}^{(l)}$  channels for any $l\in \mathcal{L}$, $f_{\Gamma_{\text{E}}^{(m^*)}}(x)=f_{\Gamma_{\text{E}}^{(m)}}(x)$ for any $m \in\mathcal{M}$. The SOP of the system is obtained by using $F_{\Gamma_{\text{U}}^{(m^*)}}(x)$ from (\ref{eq_SS_cdf}) and $f_{\Gamma_{\text{E}}^{(m^*)}}(x)$ from (\ref{eq_pdf_E}) in (\ref{eq_exact_CLT_sop_no_sel_mult_eve}) as
\begin{align}
\label{eq_SS_SOP2}
\mathcal{P}_{\textrm{out}}= 1-\sum^M_{m=1}\sum^L_{l=1}(-1)^{m+l}\binom{M}{m} \binom{L}{l} \xi^m \mathcal{P}_{\textrm{out}}^{(m,l)},
\end{align}
where 
\begin{align}
\label{eq_SS_Q}
\mathcal{P}_{\textrm{out}}^{(m,l)}=\int_0^{\infty}\Bigg[Q\Big(\frac{\sqrt{{(\rho-1+\rho x)}/{\bar{\Gamma}_{\text{U}}}}-\mu_{\textrm{U}}}{\sigma_{\text{U}}}\Big)\Bigg]^m \hat{f}_{\Gamma_{\text{E}}^{(l)}}(x).
\end{align}
When $M=1$, (\ref{eq_SS_SOP2}) provides the SOP of a single-user system.
It is difficult to derive an exact  closed-form solution of the integral in  (\ref{eq_SS_Q}), hence, we use the following approximation of the $Q$-function from \cite{paper_q_approx}, 
\begin{subequations}
  \begin{empheq}[left=Q(x)\approx\empheqlbrace]{align}
  \label{eq_Q_func_sum_3exp1}
    & \sum_{i=1}^3\frac{ w_i}{2}\exp\lb(-\frac{p_i x^2}{2}\rb) ~~~~\textrm{when} ~~x\ge0 \\
  \label{eq_Q_func_sum_3exp2}
    & 1-\sum_{i=1}^3\frac{ w_i}{2}\exp\lb(-\frac{p_i x^2}{2}\rb) \textrm{when} ~~x<0,
  \end{empheq}
\end{subequations}
where
$w_i=\lb\{\frac{1}{6}, \frac{1}{3}, \frac{1}{3}\rb\}$ and $p_i=\lb\{1, 4, \frac{4}{3}\rb\}$ for $i\in\{1,2,3\}$. To apply the $Q$-function approximation shown in (\ref{eq_Q_func_sum_3exp1})-(\ref{eq_Q_func_sum_3exp2}) in  (\ref{eq_SS_Q}), we have to divide the feasible integration region of $x$ (i.e. $0\le x\le\infty$) appropriately depending on whether the argument of $Q$-function in (\ref{eq_SS_Q}) is positive or negative in that region.
The region in which $Q$-function in (\ref{eq_SS_Q}) has a positive argument is
$\frac{\mu^2_{\text{U}}\bar{\Gamma}_{\text{U}}-\lb(\rho-1\rb)}{\rho}\le x \le\infty$ and the argument of the $Q$-function in (\ref{eq_SS_Q}) is negative in the region $0\le x<\frac{\mu^2_{\text{U}}\bar{\Gamma}_{\text{U}}-\lb(\rho-1\rb)}{\rho}$. Next, we evaluate the SOP by taking into account the following two cases.

Case 1: When $\bar{\Gamma}_{\text{U}}\le\frac{\rho-1}{\mu_{\textrm{U}}^2}$, we ascertain 
that the $Q$-function in (\ref{eq_SS_Q}) has a positive argument in the entire feasible range of $0\le x\le \infty$.
Hence, we use the approximation of the $Q$-function in (\ref{eq_Q_func_sum_3exp1}) to obtain $\mathcal{P}_{\textrm{out}}^{(m,l)}=J^{(m,l)}_+,$
where 
\begin{align}
\label{eq_J1_subopt}
J^{{(m,l)}}_+
&=\int_0^{\infty}\Bigg(\sum_{i=1}^3\frac{ w_i}{2}\exp\Bigg(-\frac{\lb(\sqrt{\frac{\rho-1+\rho x}{\bar{\Gamma}_{\text{U}}}}-\mu_{\textrm{U}}\rb)^2}{2({\sigma^{(i)}_{\text{U}}})^2}\Bigg)\Bigg)^m  \nn\\
&\times \hat{f}_{\Gamma_{\text{E}}^{(l)}}(x)dx
\end{align}
and 
$ \sigma^{(i)}_{\text{U}} =\frac{\sigma_{\text{U}} }{\sqrt{p_i}}$. 
By applying the multinomial theorem in (\ref{eq_J1_subopt})  we obtain (\ref{eq_SS_J1}),
where 
\begin{align}
\label{eq_I1_prime_subopt_ext}
\mathcal{J}_{+}^{(\boldsymbol{k},l)}((\sigma^{(\boldsymbol{k})}_{\text{U}})^2)&=
\int_0^{\infty}\frac{1}{2}\exp\Bigg(-\frac{ \lb({\sqrt{\frac{\rho-1+\rho x}{\bar{\Gamma}_{\text{U}}}}-\mu_{\textrm{U}}}\rb)^2}{2(\sigma^{(\boldsymbol{k})}_{\text{U}})^2}\Bigg) \nn\\
&\times\hat{f}_{\Gamma_{\text{E}}^{(l)}}(x)dx.
\end{align}
The solution of $\mathcal{J}_{+}^{(\boldsymbol{k},l)}((\sigma^{(\boldsymbol{k})}_{\text{U}})^2)$ is provided in Appendix \ref{appendix_Jplus_SS} and the result is expressed in (\ref{eq_ss_j_new_theorem}).

Case 2: {When $\bar{\Gamma}_{\text{U}}>\frac{\rho-1}{\mu_{\textrm{U}}^2}$, the $Q$-function in (\ref{eq_SS_Q}) has  negative  argument  in the range $0\le x<\frac{\mu^2_{\text{U}}\bar{\Gamma}_{\text{U}}-\lb(\rho-1\rb)}{\rho}$ and positive argument in  the range $\frac{\mu^2_{\text{U}}\bar{\Gamma}_{\text{U}}-\lb(\rho-1\rb)}{\rho}\le x \le\infty$.} Thus, we rewrite (\ref{eq_SS_Q}) with the help of (\ref{eq_Q_func_sum_3exp1}) and (\ref{eq_Q_func_sum_3exp2}) as
$\mathcal{P}_{\textrm{out}}^{(m,l)}=I_+^{(m,l)}+I_-^{(m,l)}$
where
\begin{align}
\label{eq_I1_subopt}
I_+^{(m,l)}
&=\int_{\frac{\mu^2_{\text{U}}\bar{\Gamma}_{\text{U}}-\lb(\rho-1\rb)}{\rho}}^{\infty}\Bigg[\sum_{i=1}^3\frac{ w_i}{2}\nn\\
&\times\exp\Bigg(-\frac{\lb(\sqrt{\frac{\rho-1+\rho x}{\bar{\Gamma}_{\text{U}}}}-\mu_{\textrm{U}}\rb)^2}{2({\sigma^{(i)}_{\text{U}}})^2}\Bigg)\Bigg]^m \hat{f}_{\Gamma_{\text{E}}^{(l)}}(x)dx,\\
\label{eq_I2_subopt}
I_-^{(m,l)}
&=\int_0^{\frac{\mu^2_{\text{U}}\bar{\Gamma}_{\text{U}}-\lb(\rho-1\rb)}{\rho}}\Bigg[1-\sum_{i=1}^3\frac{ w_i}{2} \nn\\
&\times\exp\Bigg(-\frac{\lb(\sqrt{\frac{\rho-1+\rho x}{\bar{\Gamma}_{\text{U}}}}-\mu_{\textrm{U}}\rb)^2}{2({\sigma^{(i)}_{\text{U}}})^2}\Bigg) \Bigg]^m\hat{f}_{\Gamma_{\text{E}}^{(l)}}(x)dx. 
\end{align}
For the solution of $I_+^{(m,l)}$, we apply the multinomial theorem in (\ref{eq_I1_subopt}) and after some further manipulations 
 we obtain the result in \eqref{eq_SS_I1}, where
\begin{align}
\label{eq_SS_I1_appendix2}
\mathcal{I}_{+}^{(\boldsymbol{k},l)}((\sigma^{(\boldsymbol{k})}_{\text{U}})^2)
&=
\int_{\frac{\mu^2_{\text{U}}\bar{\Gamma}_{\text{U}}-\lb(\rho-1\rb)}{\rho}}^{\infty}\exp\Bigg(\frac{- \lb({\sqrt{\frac{\rho-1+\rho x}{\bar{\Gamma}_{\text{U}}}}-\mu_{\textrm{U}}}\rb)^2}{2(\sigma^{(\boldsymbol{k})}_{\text{U}})^2}\Bigg) \nn\\
&\times\hat{f}_{\Gamma_{\text{E}}^{(l)}}(x)dx.
\end{align}
The solution of  $\mathcal{I}_{+}^{(\boldsymbol{k},l)}((\sigma^{(\boldsymbol{k})}_{\text{U}})^2)$ is  provided in  Appendix  \ref{appendix_Iplus_SS} and the result is expressed in (\ref{eq_SS_I_new}).

For the solution of $I_-^{(m,l)}$, we apply the binomial theorem in (\ref{eq_I2_subopt}) to obtain
\begin{align}
\label{eq_I2_subopt2_new}
&I_{-}^{(m,l)}
=\int_0^{\frac{\mu^2_{\text{U}}\bar{\Gamma}_{\text{U}}-\lb(\rho-1\rb)}{\rho}}\Bigg[1+
\sum_{j=1}^{m}(-1)^j\binom{m}{j} \nn\\
&\times\Bigg(\sum_{i=1}^3\frac{ w_i}{2}
\exp\Bigg(-\frac{\lb({\sqrt{\frac{\rho-1+\rho x}{\bar{\Gamma}_{\text{U}}}}-\mu_{\textrm{U}}}\rb)^2}{2({\sigma^{(i)}_{\text{U}}})^2}\Bigg)\Bigg)^j\Bigg]\hat{f}_{\Gamma_{\text{E}}^{(l)}}(x)dx\nn\\
&=\int_0^{\frac{\mu^2_{\text{U}}\bar{\Gamma}_{\text{U}}-\lb(\rho-1\rb)}{\rho}}\hat{f}_{\Gamma_{\text{E}}^{(l)}}(x)dx+\sum_{j=1}^{m}(-1)^j\binom{m}{j}(J_+^{(j,l)}-I_+^{(j,l)}).   
\end{align}
The above equation is written with the help of the integrals $J_+^{(m,l)}$ in (\ref{eq_J1_subopt}) and $I_+^{(m,l)}$ in (\ref{eq_I1_subopt}) by appropriately changing index from $m$ to $j$. The already developed solutions of $J_+^{(m,l)}$ in (\ref{eq_SS_J1}) and $I_+^{(m,l)}$ in (\ref{eq_SS_I1}) can be used to arrive at the solution of $I_-^{(m,l)}$. The solution of $I_-^{(m,l)}$ is expressed in (\ref{eq_ss_I_minus_new_theorem}).
\end{proof}

\begin{corollary}
In the case of a source with multiple antennas, the SOP of the best source antenna and user pair scheduling scheme that achieves the highest SNR among all (source antenna, user) pairs can be obtained from Theorem 1 by replacing $M$ with $K\times M$, where $K$ denotes the number of antennas at the source.
\end{corollary}
\begin{proof}
As the best source antenna-user pair  $(k^*, m^*)$ will provide the maximum antenna-user pair SNR among all
source antenna-user pairs, the received SNR at the user can be written as 
\begin{align}
\Gamma_{\text{U}}^{(k^*, m^*)}=\max_{ k, m}\{\Gamma_{\text{U}}^{(k,m)}\},
\end{align}
where the SNR at the $m$-th user from the $k$-th antenna 
$\Gamma_{\text{U}}^{(k,m)}$ is represented as
\begin{align}
\label{eq_SNR_mul_ant}
\Gamma_{\text{U}}^{(k,m)}&= \eta\bar{\Gamma}_{\text{U}}\lb\lvert\sum_{n=1}^N|h_{\text{RU}}^{(n,m)}|{h}_{\text{SR}}^{(k,n)}|\rb.\nn\\
&\lb.\times \exp(\psi^{(n,m)}+\phi^{(k,n)}+\theta^{(k,n)})\rb\rvert^2,
\end{align}
$\phi^{(k,n)}$ being the phase of the channel from the $k$-th antenna to the $n$-th RIS element and $\theta^{(k,n)}$ being the controllable RIS phase shift introduced at the $n$-th RIS element for the $k$-th antenna. 
In this case, we align the phases of $k$-th antenna to RIS channel and the RIS to $m$-th user channel pair by setting $\theta^{(k,n)}=-(\psi^{(n,m)}+\phi^{(k,n)})$. This maximizes the SNR at the scheduled user, as in the SS scheme with a single antenna at the source. 
The CDF of $\Gamma_{\text{U}}^{(k,m)}$ will be the same as that of ${\Gamma_{\text{U}}^{(m)}}$ in (\ref{eq_cdf_gamma_m_inf}), assuming independence between each end-to-end source antenna-user pair channel and thus the CDF $F_{\Gamma_{\text{U}}^{(k^*,m^*)}}(x)$  is obtained as
\begin{align}
\label{eq_SS_cdf_mul_ant}
&F_{\Gamma_{\text{U}}^{(k^*,m^*)}}(x)
=\mathbb{P}\Big[\max_{ k,  m}\{\Gamma_{\text{U}}^{(k,m)}\}\le x\Big]\nn\\
&
=\lb(1-\xi Q\Big(\frac{\sqrt{{x}/{\bar{\Gamma}_{\text{U}}}}-\mu_{\textrm{U}}}{\sigma_{\text{U}}}\Big) \rb)^{KM}~~\textrm{for~} x\ge0.
\end{align}
We note that the distribution of the eavesdropping link remains the same as in (\ref{eq_pdf_E}) for the single-antenna SS scheduling case. This is because, for any antenna selection, the fading distribution of the eavesdropping channel remains unchanged. The SOP is then derived as per Theorem 1, but with $\Gamma_{\text{U}}^{(k^*,m^*)}$ instead of $\Gamma_{\text{U}}^{(m^*)}$ used in the SOP expression.
\end{proof}

\color{black}

The theorem provides the SOP of the system for the general case which encompasses two special cases: i) a single user with a single eavesdropper and ii) a single user with multiple eavesdroppers. 
Next, we consider the user scheduling when the eavesdroppers’ CSI is available.

\section{User Scheduling with Eavesdroppers'CSI}\label{subsec_clt_sop_opt_usrsel_mult_eve}
The OS scheme can be implemented by selecting the user for which the secrecy rate is maximum when global CSI is available. This means that the CSI of the eavesdroppers is also available which can be assumed when the eavesdroppers are also a part of the network \cite{Mukherje_Swindlehurst_Principles_of_physical}. Even if the global CSI is not available in practice for selection, the performance of the OS scheme can provide the theoretical secrecy performance bound of the user scheduling schemes. 
The SOP of the OS scheme is derived as
\begin{align}
\label{eq_SOP_optimal}
\mathcal{P}_{\textrm{out}}&=\mathbb{P}\lb[\max_{m\in\mathcal{M}}\{C_S^{(m)}\}<R_{\textrm{th}}\rb].
\end{align}
As we assume  $\Gamma_{\text{U}}^{(m)}$  for each $m\in\mathcal{M}$ to be independent identically distributed and $\Gamma_{\text{E}}^{(m)}$ for each $m\in\mathcal{M}$ also to be independent identically distributed, $\mathcal{P}_{\textrm{out}}
=(\mathcal{P}_{\textrm{out}}^{(m)})^M$ where $m$ is any $m\in \mathcal{M}$. We note that $\mathcal{P}_{\textrm{out}}^{(m)}$, for any $m$, is equivalent to the SOP of a single-user system. 
Thus, the SOP for the OS scheme when eavesdropping CSI is available, is directly derived with the help of the SOP of the single-user case from Theorem \ref{theorem2} for the special case $M=1$. 
\color{black}

The SOP derived in the earlier sections provides approximate closed-form solutions; however, it is difficult to analyze the effect of system parameters on the SOP performance as these are very complex. Hence, in the next section, we will obtain approximate expressions in the high-SNR regime for the SOP of the SS and OS schemes to analyze the effect of system parameters, $P/N_0$, $N$, $M$, $d_{\text{SR}}$, $d_{\text{RU}}$, $d_{\text{RE}}$, $\rho$,  and $\eta$ on the secrecy performance.

\section{SOP in the high-SNR regime} \label{sec_asymptotic}

In this section, the SOP analysis is carried out in the high-SNR regime, i.e., when $P/N_0\approx\infty$, which implies $\bar{\Gamma}_{\text{U}}\approx \infty$ and $\bar{\Gamma}_{\text{E}}\approx \infty$. 
In this scenario, the SOP is approximated by neglecting unity from both the numerator and denominator of the ratio inside the probability in (\ref{eq_exact_CLT_sop_no_sel_mult_eve}) as
\begin{align}
\label{eq_SOP_approx}
\mathcal{P}_{\textrm{out}} &=\mathbb{P}\Bigg[\frac{\Gamma_{\text{U}}^{(m^*)}}{\Gamma_{\text{E}}^{(m^*)}}<\rho\Bigg]=\int_0^{\infty}F_{\Gamma_{\text{U}}^{(m^*)}}\lb(\rho x\rb)f_{\Gamma_{\text{E}}^{(m^*)}}(x)dx.
\end{align}
We first provide the high-SNR analysis for the single-user case, i.e., $M=1$, as it can yield simplified expressions and consequently provide a better understanding of the SOP performance with system parameters $P/N_0$, $N$, $M$, $d_{\text{SR}}$, $d_{\text{RU}}$, $d_{\text{RE}}$, $\rho$,  and $\eta$. This approach will then be generalized to the case of multi-user scheduling. 
The next section finds the SOP of the single-user scenario in the high-SNR regime.

\subsection{Single-User Scenario}\label{subsec_single_user_mult_eve_hs}

We begin with the SOP expression presented in Theorem \ref{theorem2} for $M=1$. As $M=1$,   all the indices of $m$ are dropped from (\ref{eq_sop_suboptimal_final1}) and (\ref{eq_sop_suboptimal_final2}). 
In the high-SNR regime, $\bar{\Gamma}_{\text{U}} \gg{(\rho-1)}/{\mu_{\textrm{U}}^2}$ and thus the case when $\bar{\Gamma}_{\text{U}} \le{(\rho-1)}/{\mu_{\textrm{U}}^2}$ in (\ref{eq_sop_suboptimal_final1}) is not applicable, further in (\ref{eq_ss_I_minus_new_theorem}), 
$\exp\Big(-\frac{\mu^2_{\text{U}}\bar{\Gamma}_{\text{U}}-\lb(\rho-1\rb)}{\rho\lambda_{\text{E}}^{(l)}}\Big)\approx\exp\Big(-\frac{\mu^2_{\text{U}}\bar{\Gamma}_{\text{U}}}{\rho\lambda_{\text{E}}^{(l)}}\Big)$.  Moreover, as $\mu_{\textrm{U}}^2\bar{\Gamma}_{\text{U}} \approx\infty$, $I_{+}^{(l)}$ is negligible as compared to $J_{+}^{(l)}$ in (\ref{eq_ss_I_minus_new_theorem}) due to the negligible integration region in (\ref{eq_I1_subopt}) for $I_{+}^{(l)}$.  
 Thus, the SOP from (\ref{eq_sop_suboptimal_final2}) can be approximated by
\begin{align}
\label{eq_sop_high0}
\mathcal{P}_{\textrm{out}}=
1-\sum^L_{l=1}(-1)^{l+1}\binom{L}{l}\xi\Big(1-\exp\Big(-\frac{\mu^2_{\text{U}}\bar{\Gamma}_{\text{U}}}{\rho\lambda_{\text{E}}^{(l)}}\Big)-J_{+}^{(l)} \Big),
\end{align}
where $J_{+}^{(l)}=\sum_{i=1}^3 w_i\mathcal{J}_{+}^{(i,l)}(({\sigma^{(i)}_{\text{U}}})^2)$ is obtained from  (\ref{eq_J1_subopt})  when $M=1$. Subsequently, to obtain $\mathcal{J}_{+}^{(i,l)}(({\sigma^{(i)}_{\text{U}}})^2)$ in the high-SNR regime, we approximate $(\rho-1)/\bar{\Gamma}_\text{U}\approx0$ in (\ref{eq_I1_prime_subopt_ext}) as $\bar{\Gamma}_\text{U}\approx\infty$. Thus,  $\mathcal{J}_{+}^{(i,l)}(({\sigma^{(i)}_{\text{U}}})^2)$ in (\ref{eq_I1_prime_subopt_ext}) is approximated in the high-SNR regime as
\begin{align}
\label{eq_Jplus_approx}
&\mathcal{J}_{+}^{(i,l)}(({\sigma^{(i)}_{\text{U}}})^2)=\int_0^{\infty}\frac{1}{2}\exp\Big(-\frac{\lb(\sqrt{\frac{\rho x}{\bar{\Gamma}_{\text{U}}}}-\mu_{\textrm{U}}\rb)^2}{2({\sigma^{(i)}_{\text{U}}})^2} \Big)
\hat{f}_{\Gamma_{\text{E}}^{(l)}}(x)dx
\\
\label{eq_Jplus_closed}
&=\frac{\bar{\Gamma}_{\text{U}}}{2\rho\lambda_{\text{E}}^{(l)}\Upsilon^{(i,l)}}
\Bigg[\exp\Big(\frac{-\mu^2_{\text{U}}}{2({\sigma^{(i)}_{\text{U}}})^2}\Big)+\frac{\mu_{\textrm{U}}\sqrt{\pi}}{({\sigma^{(i)}_{\text{U}}})^2\sqrt{\Upsilon^{(i,l)}}}\nn \\
&\times \exp\Big(\frac{-\mu^2_{\text{U}}\bar{\Gamma}_{\text{U}}}{2({\sigma^{(i)}_{\text{U}}})^2 \rho\lambda_{\text{E}}^{(l)}
\Upsilon^{(i,l)}}
\Big)
Q\Big(\frac{-\mu_{\textrm{U}}\sqrt{2}}{2({\sigma^{(i)}_{\text{U}}})^2\sqrt{\Upsilon^{(i,l)}}}\Big) 
\Bigg].
\end{align}
The proof of the solution of $\mathcal{J}_{+}^{(i,l)}(({\sigma^{(i)}_{\text{U}}})^2)$ obtained in (\ref{eq_Jplus_closed}) is provided in Appendix \ref{appendix_proof_J_plus}. Therefore,  we obtain the SOP in the high-SNR regime with the help of   (\ref{eq_sop_high0}) and (\ref{eq_Jplus_closed}).  

As $N$ is very large in a typical RIS-aided system,  $\frac{\mu^2_{\text{U}}\bar{\Gamma}_{\text{U}}}{\rho\lambda_{\text{E}}^{(l)}}$ and $\frac{\mu^2_{\text{U}}}{({\sigma^{(i)}_{\text{U}}})^2}$ are very large and 
$\Upsilon^{(i,l)}=\frac{1}{2({\sigma^{(i)}_{\text{U}}})^2}+\frac{ \bar{\Gamma}_{\text{U}}}{\rho\lambda_{\text{E}}^{(l)}}$ is very small. 
Hence, we can approximate $\xi\approx1$, $\exp\Big(\frac{-\mu^2_{\text{U}}\bar{\Gamma}_{\text{U}}}{\rho\lambda_{\text{E}}^{(l)}}\Big)\approx0$,
$\exp\Big(\frac{-\mu^2_{\text{U}}}{2({\sigma^{(i)}_{\text{U}}})^2}\Big)\approx0$,  and 
$Q\Big(\frac{-\mu_{\textrm{U}}\sqrt{2}}{2({\sigma^{(i)}_{\text{U}}})^2\sqrt{\Upsilon^{(i,l)}}}\Big)\approx1$. Thus,  $\mathcal{P}_{\textrm{out}}$ in \eqref{eq_sop_high0} is approximated by
\begin{align}
\label{eq_sop_hsnr2}
\mathcal{P}_{\textrm{out}}
&=\sum^L_{l=1}(-1)^{l+1}\binom{L}{l}\sum_{i=1}^3
\frac{w_i\bar{\Gamma}_{\text{U}} \mu_{\textrm{U}}\sqrt{\pi}}{2\rho\lambda_{\text{E}}^{(l)}({\sigma^{(i)}_{\text{U}}})^2\Upsilon^{(i,l)}\sqrt{\Upsilon^{(i,l)}}} \nn\\
&\times\exp\Big(\frac{-\mu^2_{\text{U}}\bar{\Gamma}_{\text{U}}}{2({\sigma^{(i)}_{\text{U}}})^2 \rho\lambda_{\text{E}}^{(l)}
\Upsilon^{(i,l)}}
\Big).
\end{align}
To obtain the SOP directly with respect to the system parameters, we substitute  $\bar{\Gamma}_{\text{U}}=\frac{P\zeta_{\text{SU}}}{N_0}$,  $\lambda_{\text{E}}^{(l)}=N\bar{\Gamma}_{\text{E}}\eta^2/l$, $\mu_{\textrm{U}}= {\eta N\pi}/{4}$,  and $\sigma_{\text{U}}=\eta \sqrt{N({16-\pi^2})/{16}}$ into (\ref{eq_sop_hsnr2}), to get
\begin{align}
\label{eq_sop_hsnr3}
&\mathcal{P}_{\textrm{out}}
=\sum^L_{l=1}(-1)^{l+1}\binom{L}{l}
\sum_{i=1}^3\exp\Bigg(\frac{-N\pi^2}
{16\Big(\frac{\rho}{\frac{l\zeta_{\text{SU}}}{ \zeta_{\text{SE}}}}+\frac{2}{p_i}\big(\frac{16-\pi^2}{16}\big)\Big)}\Bigg)\nn\\
&\times\frac{w_i\lb(\frac{16-\pi^2}{16}\rb)\pi\sqrt{2\pi\rho N p_i}}
{\lb(\frac{\rho}{\frac{l\zeta_{\text{SU}}}{p_i\zeta_{\text{SE}}}}+2\lb(\frac{16-\pi^2}{16}\rb)\rb)
\sqrt{\lb(16-\pi^2\rb)\lb(\rho+\frac{2l\zeta_{\text{SU}}}{p_i\zeta_{\text{SE}}}\lb(\frac{16-\pi^2}{16}\rb)\rb)}}.
\end{align}
%
As $N$ is very large, the exponential factor in \eqref{eq_sop_hsnr3} is very small. As $l$ increases, this factor further decreases; hence, the contribution of the terms from $l=2$ to $l=L$ in the summation in \eqref{eq_sop_hsnr3} is considerably smaller as compared to the contribution of the term $l=1$ alone. Hence, we retain only $l=1$ to obtain the approximate SOP in the high-SNR regime as
\begin{align}
\label{eq_sop_hsnr_final}
\mathcal{P}_{\textrm{out}}
&\approx L
\sum_{i=1}^3
\exp\Bigg(-\frac{N\pi^2}
{16\lb(\frac{\rho}{\frac{\zeta_{\text{SU}}}{ \zeta_{\text{SE}}}}+\frac{2}{p_i}\lb(\frac{16-\pi^2}{16}\rb)\rb)}\Bigg)
\times\nn \\
&\frac{w_i\lb(\frac{16-\pi^2}{16}\rb)\pi\sqrt{2\pi\rho N p_i}}
{\lb(\frac{\rho}{\frac{\zeta_{\text{SU}}}{p_i\zeta_{\text{SE}}}}+2\lb(\frac{16-\pi^2}{16}\rb)\rb)
\sqrt{\lb(16-\pi^2\rb)\lb(\rho+\frac{2\zeta_{\text{SU}}}{p_i\zeta_{\text{SE}}}\lb(\frac{16-\pi^2}{16}\rb)\rb)}}.
\end{align}
The summation over $i$ in  (\ref{eq_sop_hsnr_final}) is due to the  approximation of  $Q$-function used in (\ref{eq_Q_func_sum_3exp1})-(\ref{eq_Q_func_sum_3exp2}).

We can now make several important observations and key insights from (\ref{eq_sop_hsnr_final}), which was not possible earlier from the exact SOP expression.

\begin{remark}
 {In the high-SNR regime, the SOP saturates to a constant value that is independent of $\eta$ and $\xi$.}
\end{remark}
\begin{remark}{It can be observed from \eqref{eq_sop_hsnr_final} that the SOP increases linearly with the number of eavesdroppers in the high-SNR regime.} 
\end{remark}
\begin{remark}{We also find that the SOP is proportional to $\sqrt{N}\exp(-\beta N)$, where $\beta=\frac{\pi^2\frac{\zeta_{\text{SU}}}{\zeta_{\text{SE}}}}
{16\lb(\rho+\frac{2\zeta_{\text{SU}}}{p_i\zeta_{\text{SE}}}\lb(\frac{16-\pi^2}{16}\rb)\rb)}$.
We note that $\beta>0$ and thus for large value of $N$, $\sqrt{N} \exp(-\beta N) = \exp(-(\beta N - 1/2 \ln N)) \approx \exp(-\beta N)$. 
This shows that the SOP decreases exponentially with the number of RIS elements when this number is large.} 
\end{remark}
\color{black}

Now, let us consider how RIS and other node locations affect the SOP performance given by (\ref{eq_sop_hsnr_final}). 
 We notice that the SOP depends on the ratio $\frac{\zeta_{\text{SU}}}{\zeta_{\text{SE}}}$ which is the ratio of the path losses for the indirect links from the source to the users and from the source to the eavesdroppers. This is reasonable as this is the factor that affects the signal strength at the user relative to the signal strength at the eavesdropper, which is essentially the measure of secrecy. 

\begin{remark}
 The ratio $\frac{\zeta_{\text{SU}}}{\zeta_{\text{SE}}}$ can be written
using (\ref{path_loss_angle}) as 
$\frac{\zeta_{\text{SU}}}{\zeta_{\text{SE}}}=\lb(\frac{\cos{\gamma^{(\text{inc})}+\cos{\gamma^\text{(ref)}_{\text{U}}}}}{\cos{\gamma^{(\text{inc})}+\cos{\gamma^\text{(ref)}_{\text{E}}}}}\rb)^2\frac{d_{\text{RE}}^2}{d_{\text{RU}}^2}$. This ratio along with (\ref{eq_sop_hsnr_final}) provides the SOP  performance as a function of the angles $\gamma_\text{inc}$, $\gamma^\text{(ref)}_{\text{U}}$, $\gamma^\text{(ref)}_{\text{E}}$, and the ratio $d_{\text{RE}}/d_{\text{RU}}$, which can be used to determine the optimal placement of the RIS.   
\end{remark}
\color{black}

 In conclusion, it can be stated that the high-SNR closed-form SOP expression in  (\ref{eq_sop_hsnr_final}) provides a useful tool for performance prediction and system design, such as how many elements are required and where to place the RIS for an acceptable SOP performance. This is otherwise difficult to achieve through a very time-consuming simulation method, especially when the number of RIS elements is large.

\subsection{User Scheduling without Eavesdroppers' CSI}
\label{subsec_subopt_usersel_hsnr}
To find the SOP in the high-SNR regime for the SS scheme, we start with the SOP expression derived in Theorem \ref{theorem2} as in the previous section. As $\bar{\Gamma}_{\text{U}} \gg(\rho-1)/\mu_{\textrm{U}}^2$ in the high-SNR regime, (\ref{eq_sop_suboptimal_final1}) does not apply. As $N$ is usually very large,  $\xi\approx1$. Hence, after performing some algebraic manipulations on (\ref{eq_sop_suboptimal_final2}) in the high-SNR regime, we write $\mathcal{P}_{\textrm{out}}=1-\mathcal{P}_1-\mathcal{P}_2-\mathcal{P}_3$
where
\begin{align}
\label{eq_SS_P1}
\mathcal{P}_1&
=1+\sum^L_{l=1}(-1)^{l} \binom{L}{l}e^{-\lb(\frac{\mu^2_{\text{U}}\bar{\Gamma}_{\text{U}}}{\rho\lambda_{\text{E}}^{(l)}}\rb)},\\
\label{eq_SS_P2_first}
\mathcal{P}_2
&=\sum^M_{m=1}\sum_{j=1}^{m}(-1)^{m+j}\binom{M}{m}\binom{m}{j}\nn\\
&\times\Bigg(\sum^L_{l=1}(-1)^{l} \binom{L}{l}\Big(J_+^{(j,l)}-I_+^{(j,l)}\Big)\Bigg)\\
 \label{eq_SS_P3}
\mathcal{P}_3&
=\sum^L_{l=1}(-1)^{l}\binom{L}{l}\Bigg( \sum^M_{m=1} \binom{M}{m}(-1)^{m}I_+^{(m,l)}\Bigg).
\end{align}
It is to be noted that $\mathcal{P}_1$ and $\mathcal{P}_2$ are due to the substitution of $I_-^{(m,l)}$ from (\ref{eq_ss_I_minus_new_theorem}) into (\ref{eq_sop_suboptimal_final2}) and $\mathcal{P}_3$ is due to $I_+^{(m,l)}$ in (\ref{eq_sop_suboptimal_final2}). Specifically, $\mathcal{P}_1$ results from the first part of $I_-^{(m,l)}$ in (\ref{eq_ss_I_minus_new_theorem}) that does not include summation terms and $\mathcal{P}_2$ results from the second part of $I_-^{(m,l)}$ in (\ref{eq_ss_I_minus_new_theorem}) that includes summation terms. $\mathcal{P}_2$ can be further simplified as 
\begin{align}
\label{eq_SS_P2}
\mathcal{P}_2
&=\sum^M_{m=1} \Bigg(\sum^L_{l=1}(-1)^{l} \binom{L}{l}\Big(J_+^{(m,l)}-I_+^{(m,l)}\Big)\Bigg)\nn\\
&\times\Bigg( \sum_{j=m}^{M}(-1)^{m+j} \binom{M}{j}\binom{j}{m}\Bigg)\nn\\
&=\sum^L_{l=1}(-1)^{l} \binom{L}{l}\Big(J_+^{(M,l)}-I_+^{(M,l)}\Big),
\end{align}
where we use that fact that $\sum_{j=m}^{M}(-1)^{m+j} \binom{M}{j}\binom{j}{m}=1$ only when $m=M$ and $j=M$, (otherwise it is zero).  Finally, $\mathcal{P}_{out}$ 
is written as
\begin{align}\label{eq_SS_asym_breakups_T}
\mathcal{P}_{out}&=\sum^L_{l=1}(-1)^{(l+1)} \binom{L}{l}\Big(\mathcal{P}_1^{(l)}+\mathcal{P}_2^{(l)}+\mathcal{P}_3^{(l)}\Big)
\end{align}
where $\mathcal{P}_1^{(l)}=e^{-\big(\frac{\mu^2_{\text{U}}\bar{\Gamma}_{\text{U}}}{\rho\lambda_{\text{E}}^{(l)}}\big)}$, $\mathcal{P}_2^{(l)}=J_+^{(M,l)}-I_+^{(M,l)}$, and $\mathcal{P}_3^{(l)}=\sum^M_{m=1} \binom{M}{m}(-1)^{m}I_+^{(m,l)}$.
We note from (\ref{eq_SS_asym_breakups_T}) that the SOP in the high-SNR regime requires $J_+^{(M,l)}$, $I_+^{(M,l)}$, and $I_+^{(m,l)}$ defined in (\ref{eq_SS_J1}) and (\ref{eq_SS_I1}) in the high-SNR regime. Consequently,  we evaluate $\mathcal{J}_{+}^{(\boldsymbol{k},l)}((\sigma^{(\boldsymbol{k})}_{\text{U}})^2)$ and $\mathcal{I}_{+}^{(\boldsymbol{k},l)}((\sigma^{(\boldsymbol{k})}_{\text{U}})^2)$ in the high-SNR regime.
 The solution of $\mathcal{J}_{+}^{(\boldsymbol{k},l)}((\sigma^{(\boldsymbol{k})}_{\text{U}})^2)$ in the high-SNR regime is already obtained in (\ref{eq_Jplus_closed}). In (\ref{eq_Jplus_closed}), the index $i$ must be replaced with $\boldsymbol{k}$ for the solution. 
Next, $\mathcal{I}_{+}^{(\boldsymbol{k},l)}((\sigma^{(\boldsymbol{k})}_{\text{U}})^2)$ defined in (\ref{eq_SS_I1_appendix2}) is adapted to the high-SNR regime with the approximation of $\bar{\Gamma}_{\text{U}}\gg(\rho-1)/\mu^2_{\text{U}}$ and $(\rho-1)/\bar{\Gamma}_\text{U}\approx 0$ due to $\bar{\Gamma}_\text{U}\approx\infty$ as
\begin{align}
&\label{eq_Iplus_high_SNR}
\mathcal{I}_{+}^{(\boldsymbol{k},l)}((\sigma^{(\boldsymbol{k})}_{\text{U}})^2)
=
\int_{\frac{\mu^2_{\text{U}}\bar{\Gamma}_{\text{U}}}{\rho}}^{\infty}\exp\Bigg(-\frac{ \Big({\sqrt{\frac{\rho x}{\bar{\Gamma}_{\text{U}}}}-\mu_{\textrm{U}}}\Big)^2}{2(\sigma^{(\boldsymbol{k})}_{\text{U}})^2}\Bigg)\nn\\
&\times \hat{f}_{\Gamma_{\text{E}}^{(l)}}(x)dx\\
\label{eq_Iplus_high_SNR_closed}
&=\frac{\bar{\Gamma}_{\text{U}}}{2\rho\lambda_{\text{E}}^{(l)}\Upsilon^{(\boldsymbol{k},l)}}
\Bigg[\exp\Big(\frac{-\mu^2_{\text{U}} \bar{\Gamma}_{\text{U}}}{\rho\lambda_{\text{E}}^{(l)}}\Big)+\frac{\mu_{\textrm{U}}\sqrt{\pi}}{({\sigma^{(\boldsymbol{k})}_{\text{U}}})^2\sqrt{\Upsilon^{(\boldsymbol{k},l)}}} \nn\\
&\times\exp\Big(\frac{-\mu^2_{\text{U}} \bar{\Gamma}_{\text{U}}}
{2({\sigma^{(\boldsymbol{k})}_{\text{U}}})^2\rho\lambda_{\text{E}}^{(l)}\Upsilon^{(\boldsymbol{k},l)}}\Big)Q\Big(\frac{\sqrt{2}\mu_{\textrm{U}} \bar{\Gamma}_{\text{U}}}{\rho\lambda_{\text{E}}^{(l)}\sqrt{\Upsilon^{(\boldsymbol{k},l)}}
}\Big)
\Bigg].
\end{align}
The proof of the solution in (\ref{eq_Iplus_high_SNR_closed}) is provided in  Appendix \ref{appendix_proof_I_plus}. 




In (\ref{eq_SS_asym_breakups_T}) as $M$ increases,  both $J_+^{(M,l)}$ and  $I_+^{(M,l)}$ decrease in $\mathcal{P}_2^{(l)}$.
This can be verified from the definitions of $J_+^{(m,l)}$ and  $I_+^{(m,l)}$ in (\ref{eq_J1_subopt}) and (\ref{eq_I1_subopt}), respectively, when $m=M$. $J_+^{(m,l)}$ and  $I_+^{(m,l)}$ are nothing but averaging $(Q(\cdot))^m$ after applying the approximation to $Q$-function from \eqref{eq_Q_func_sum_3exp1}. 
Thus, the values of both $J_+^{(m,l)}$ and  $I_+^{(m,l)}$ are positive and less than unity when $m=1$. When $m=M$, the values of $J_+^{(M,l)}$ and  $I_+^{(M,l)}$ decrease with $M$ as $m$ is in the exponent of $Q$-function. In contrast, the absolute value of $\mathcal{P}_3^{(l)}$ 
increases with  $M$ in (\ref{eq_SS_asym_breakups_T}).
The term $\mathcal{P}_3^{(l)}$ is nothing but averaging $((1-Q(\cdot))^M-1)$
over a certain region, as can be verified from (\ref{eq_I1_subopt}). As $(1-Q(\cdot))^M$ decreases as $M$ increases, the absolute value of $\mathcal{P}_3^{(l)}$ increases. Here we note that $\mathcal{P}_3^{(l)}$ is negative as $\lb((1-Q(\cdot))^M-1\rb)$ is negative. By comparing $\mathcal{P}_2^{(l)}$ and $\mathcal{P}_3^{(l)}$, we ignore $\mathcal{P}_2^{(l)}$ in \eqref{eq_SS_asym_breakups_T} since its  contribution is negligible as compared to $\mathcal{P}_3^{(l)}$ when $M$ is large. Therefore, we obtain
\begin{align}
\label{eq_SS_approx_hsnr}
\mathcal{P}_{\textrm{out}}
&=\sum^L_{l=1}(-1)^{(l+1)}  \binom{L}{l}\Bigg[e^{-\lb(\frac{\mu^2_{\text{U}}\bar{\Gamma}_{\text{U}}}{\rho\lambda_{\text{E}}^{(l)}}\rb)}\nn\\
&+\sum^M_{m=1} \binom{M}{m}(-1)^{m}I_+^{(m,l)}\Bigg].
\end{align}
Next, by taking the dominant term $l=1$ in the summation series of (\ref{eq_SS_approx_hsnr}) due to the same reasoning provided in the case of the single-user scenario in (\ref{eq_sop_hsnr_final}), we obtain  the SOP expression in the high-SNR regime as \begin{align}
   \label{eq_SS_asymp}
 \mathcal{P}_{\textrm{out}}\approx L\Bigg[ \exp{\lb(-\frac{\mu^2_{\text{U}}\bar{\Gamma}_{\text{U}}}{\rho{\lambda_{\text{E}}^{(1)}}}\rb)}+\sum^M_{m=1}(-1)^{m}\binom{M}{m} I_+^{(m,1)}\Bigg].
\end{align}
 By substituting the high-SNR version of $I_+^{(m,1)}$ with the help of  (\ref{eq_Iplus_high_SNR_closed}), 
  $\bar{\Gamma}_{\text{U}}=\frac{P\zeta_{\text{SU}}}{N_0}$, ${\lambda_{\text{E}}^{(1)}}=\eta^2N\bar{\Gamma}_{\text{E}}$, $\mu_{\textrm{U}}= {\eta N\pi}/{4}$, and $\sigma_{\text{U}}=\eta \sqrt{N\lb({16-\pi^2}\rb)/{16}}$   in (\ref{eq_SS_asymp}); and thereafter performing some algebraic manipulations by replacing the $Q$-function  with its approximation in (\ref{eq_Q_func_sum_3exp1}), we obtain the SOP expression in the high-SNR regime in terms of the basic system parameters as
 \begin{align}\label{eq_SS_asymp_final}
&\mathcal{P}_{\textrm{out}}\approx L\Bigg[\Bigg(1
+\sum^M_{m=1}
\frac{W^{(m)}\mathcal{A} \frac{\zeta_{\text{SU}}}{\zeta_{\text{SE}}}}
{\lb(\rho+2\mathcal{A} \frac{\zeta_{\text{SU}}}{\zeta_{\text{SE}}}\rb)}
 \Bigg)\exp{\lb(-\frac{N\pi^2\zeta_{\text{SU}}}{16\rho \zeta_{\text{SE}}}\rb)} \nn\\
 &+\sum^M_{m=1}\sum_{i=1}^3
\frac{W^{(m)}\mathcal{A} w_i \frac{\zeta_{\text{SU}}}{\zeta_{\text{SE}}}}
{2\lb(\rho+2\mathcal{A} \frac{\zeta_{\text{SU}}}{\zeta_{\text{SE}}}\rb)}\frac{\pi\sqrt{2\pi\rho N \sum^3_{i=1}k_ip_i}}
{\sqrt{\lb(16-\pi^2\rb)\lb(\rho+2\mathcal{A} \frac{\zeta_{\text{SU}}}{\zeta_{\text{SE}}}\rb)}}\nn\\
&\times \exp\Bigg(-\frac{N\pi^2\frac{\zeta_{\text{SU}}}{\zeta_{\text{SE}} }\lb(\frac{1}{2}+p_i\mathcal{A} \frac{\zeta_{\text{SU}}}{\rho\zeta_{\text{SE}}}\rb)}
{8\lb(\rho+2\mathcal{A} \frac{\zeta_{\text{SU}}}{\zeta_{\text{SE}}}\rb)}\Bigg)\Bigg],
\end{align}
where $W_m=(-1)^{m}\binom{M}{m}  \sum_{\boldsymbol{k} \in \mathcal{S}_m}\binom{m}{\boldsymbol{k}}
\frac{w_1^{k_1}w_2^{k_2}w_3^{k_3}}{2^{k_1+k_2+k_3-1}}$
and  $\mathcal{A}=\big({\frac{16-\pi^2}{16}}\big)\big{/}{\sum^3_{i=1} k_ip_i}$.


We now deduce the behavior of the SOP with the system parameters.  We observe from  (\ref{eq_SS_asymp}) that the first term is independent of $M$, and the second term depends on $M$. The second term decreases with $M$ as mentioned earlier when we neglected $\mathcal{P}_3^{(l)}$ in (\ref{eq_SS_asym_breakups_T}). 
Therefore, the SS scheme will benefit from the increasing number of users.
{
\begin{remark}
It can be shown that the asymptotic SOP  is a decreasing function of $M$. 
Therefore, the secrecy performance improves with an increasing number of users.
\end{remark}}
We also notice from (\ref{eq_SS_asymp_final}) that the SOP in the high-SNR regime increases linearly with $L$ and is independent of $\eta$ and $\xi$.
From (\ref{eq_SS_asymp_final}), we conclude that the SOP saturates in the high-SNR regime.

{
\begin{remark}
If we consider the SOP performance with $N$, we observe that the first term in (\ref{eq_SS_asymp_final}) is proportional to $\exp(-\beta_1N)$ and the second term is proportional to $\sqrt{N}\exp(-\beta_2N)$, where $\beta_1=\frac{\pi^2\zeta_{\text{SU}}}{16\rho \zeta_{\text{SE}}}>1$ and $\beta_2=\frac{\pi^2\frac{\zeta_{\text{SU}}}{\zeta_{\text{SE}} }\lb(\frac{1}{2}+p_i\mathcal{A} \frac{\zeta_{\text{SU}}}{\rho\zeta_{\text{SE}}}\rb)}
{8\lb(\rho+2\mathcal{A} \frac{\zeta_{\text{SU}}}{\zeta_{\text{SE}}}\rb)}>1$. For large values of $N$, $\sqrt{N} \exp(-\beta_2 N) = \exp(-(\beta_2 N - 1/2 \ln N)) \approx \exp(-\beta_2 N)$. Hence, we conclude that both of the terms in \eqref{eq_SS_asymp_final} decrease exponentially with the number of RIS elements when this number is large, as was explored for the single-user case in Section \ref{subsec_single_user_mult_eve_hs}.
\end{remark} 
The reasoning for the above {Remark} is also similar to that in the single-user scenario.} 
The behavior with $N$, $L$, $P/N_0$, and path loss is similar to the single-user scenario.
{
\begin{remark}
The   SOP in the high-SNR regime is a monotonically decreasing function of $\frac{\zeta_{\text{SU}}}{\zeta_{\text{SE}}}$ and its behavior is similar to that in the single-user case; hence, the effect of the RIS location is the same as in the single-user case.      
\end{remark}
}


In the high-SNR regime, the SOP expression derived for the single-user case is not the special case of the SOP expression derived for the SS scheme, though the SOP of the single-user case is a special case of the SOP of the SS scheme.  We arrive at the single-user and multi-user SOP expressions by adopting different approximations from a general SOP expression of the SS scheme in the high-SNR regime. In particular, when $M=1$, $I_+^{(l)}=I_+^{(1,l)}$ is negligible as compared to $J_+^{(l)}=J_+^{(1,l)}$ in (\ref{eq_sop_high0}), however, as $M$ increases, $\mathcal{P}_2^{(l)}$ becomes negligible as compared to $\mathcal{P}_3^{(l)}$ in (\ref{eq_SS_approx_hsnr}).

\subsection{User Scheduling with Eavesdroppers' CSI}
\label{subsec_opt_hsnr}
The SOP of the OS scheme in the high-SNR regime is obtained with the help of the relationship between the SOP of the OS scheme and the SOP of the single-user system provided in (\ref{eq_SOP_optimal}) using the single-user approximate SOP in  (\ref{eq_sop_hsnr_final}). 
{
\begin{remark}
    As $M$ appears in the exponent of (\ref{eq_SOP_optimal}) for the SOP of the OS scheme, we conclude that the SOP in the OS scheme decreases at a rate of $\exp(-MN)$ with  $M$ and $N$ and increases with the increase in $L$ at a rate of $L^M$ in the high-SNR regime when $N$ and $M$ are large. 
\end{remark}
} 



\section{DF Relay-aided Scheduling} \label{sec_relay_based}
In this section, we compare the RIS-aided scheduling with the DF relay-aided scheduling.
We derive the SOP of a system where a relay decodes the received signal in the first time slot, and re-encodes and forwards it to the users in the next time slot. 
R denotes the relay instead of the RIS in this section.  
We consider two situations with relay-aided scheduling: i)  without direct links and ii)  with direct links between the source and the users. The direct links are incorporated in the relaying to assess how many RIS elements are required to beat the relaying system with an advantage of the direct links even though the RIS-aided system does not have that advantage. 
In both cases, the user is selected as in the RIS-aided SS scheme, i.e., based on the relay to user channel SNRs only. In case ii), we assume maximum ratio combining (MRC) at the scheduled user after two time slots.

As the rate of a dual-hop link with a DF relay is limited by the rate of the worst hop, the SNR for the scheduled user U when the  direct link is present is given by
\begin{align}
\label{eq_snr_df}
\Gamma_{\text{U}}
=\frac{P_1\lvert {h_{\text{SU}}} \rvert^2}{N_0}+\min\Bigg(\frac{P_1\lvert {h_{\text{SR}}} \rvert^2}{N_0},  \frac{P_2\max\limits_{m\in\mathcal{M}}\lb\{\lvert  {h}_{\text{RU}}^{(m)}\rvert^2\rb\}}{N_0}\Bigg),
\end{align}
where $P_1$ and $P_2$ are the transmit powers of S and R, respectively,  {${h}_{\text{XY}}=\zeta_{\text{XY}}\tilde{h}_{\text{SU}}$ is the direct channel between the nodes X and Y where $\text{XY}\in\{\text{SU}, \text{SR}, \text{RU}\}$,  $\tilde{h}_{\text{XY}}$ is a complex Gaussian RV with zero mean and unit variance,  $\zeta_{\text{XY}} = ((4\pi/\lambda)^2d_{\text{XY}}^{\upsilon})^{-1}$ is the path loss between the nodes X and Y, and $\upsilon$ is the path loss exponent.
We note that the path loss model assumed in the relay-aided system is the free space path loss for the direct link between nodes, which is different to the free space path loss for the indirect links considered in the RIS-aided system \cite{nemanja_TWC}.
The definition of the channels S-R, R-$\text{E}^{(l)}$ for each $l$,  and R-$\text{U}^{(m)}$ for each $m$ are the same as in the RIS-aided system.}
We assume the same average transmit power for the RIS and relay-aided systems, i.e., $P={(P_1+P_2)}/{2}$.
The achievable SNR for the eavesdropping channel is
 $\Gamma_{\text{E}}
={P_2\max_{l\in\mathcal{L}}\{\lvert {h}_{\text{RE}}^{(l)}\rvert^2\}}/{N_0}$  
where ${h}_{\text{RE}}^{(l)}=\zeta_{\text{RE}} \tilde{h}_{\text{RE}}^{(l)}$
is the direct channel between the nodes  $\text{R}$ and $\text{E}^{(l)}$, $\tilde{h}_{\text{RE}}^{(l)}$ is a complex Gaussian RV  with zero mean and unit variance, and $\zeta_{\text{RE}} = ((4\pi/\lambda)^2d_{\text{RE}}^{\upsilon})^{-1}$ is the path loss for the distance $d_{\text{RE}}$ between the nodes R and $\text{E}^{(l)}$ \cite{kundu_dual_hop_regenerative}.

To obtain the SOP, we consider the definition of the SOP for DF relaying from \cite{kundu_dual_hop_regenerative} and derive it including the direct link, multiple eavesdroppers, and user scheduling as per the system assumed here. We obtain the distributions of $\Gamma_{\text{U}}$ and $\Gamma_{\text{E}}$ to find the SOP using (\ref{eq_exact_CLT_sop_no_sel_mult_eve}). The CDF of $\Gamma_{\text{U}}$ is obtained as
\begin{align}
\label{eq_cdf_gamma_SRD}
&F_{\Gamma_{\text{U}} }(x)
=1-\exp\Big(-\frac{x}{\lambda_{\text{SU}}}\Big)-\sum^M_{m=1}\frac{(-1)^{m+1}\binom{M}{m}}{1-\frac{\lambda_{\text{SU}}}{\lambda_{\text{SR}}}-\frac{\lambda_{\text{SU}}}{\lambda_{\text{RU}}^{(m)}}}\nn\\
&\times\Bigg[\exp\Big(-x\Big(\frac{1}{\lambda_{\text{SR}}}+\frac{1}{\lambda_{\text{RU}}^{(m)}}\Big)\Big)-\exp\Big(-\frac{x}{\lambda_{\text{SU}}}\Big)\Bigg],
\end{align}
where $\lambda_{\text{SU}}=  2\alpha\Gamma_0\zeta_{\text{SU}}$, $\lambda_{\text{SR}}=  2\alpha\Gamma_0\zeta_{\text{SR}}$, $\lambda_{\text{RU}}= 2(1-\alpha)\Gamma_0\zeta_{\text{RU}}$, $\lambda_{\text{RU}}^{(m)}=\lambda_{\text{RU}}/m$, and $\Gamma_0=P/N_0$ with the power allocation factors for the source and the relay $\alpha$ and $(1-\alpha)$, respectively, wherein $0<\alpha<1$. 
The PDF of $\Gamma_{\text{E}}$, is the same as  (\ref{eq_pdf_E}) with  $\lambda_{\text{E}}= 2(1-\alpha)\Gamma_0\zeta_{\text{RE}}$   and $\lambda_{\text{E}}^{(l)}=\lambda_{\text{E}}/l$. 
By substituting \eqref{eq_cdf_gamma_SRD} and (\ref{eq_pdf_E}) in \eqref{eq_exact_CLT_sop_no_sel_mult_eve}, we obtain the SOP as 
\begin{align}
\label{eq_sop_single_rel_integral}
\mathcal{P}_{\textrm{out}}
&=1-\sum^M_{m=1}\sum^L_{l=1}\binom{M}{m}\binom{L}{l}\frac{(-1)^{m+l}}{\lambda_{\text{RE}}^{(l)}}\Bigg[\frac{\exp\lb(-\frac{\lb(\rho-1\rb)}{\lambda_{\text{SU}}} \rb)}{\frac{\rho}{\lambda_{\text{SU}}}+\frac{1}{\lambda_{\text{RE}}^{(l)}}}
\nn\\
&+
\frac{\exp\Big(-\lb(\rho-1\rb)\lb(\frac{1}{\lambda_{\text{SR}}}+\frac{1}{\lambda_{\text{RU}}^{(m)}} \rb)\Big)}{\lb(1-\frac{\lambda_{\text{SU}}}{\lambda_{\text{SR}}}-\frac{\lambda_{\text{SU}}}{\lambda_{\text{RU}}^{(m)}}\rb)\lb(\frac{\rho}{\lambda_{\text{SR}}}-\frac{\rho}{\lambda_{\text{RU}}^{(m)}}-\frac{1}{\lambda_{\text{RE}}^{(l)}}\rb)} \Bigg],
\end{align}
where $\rho=2^{2R_{\textrm{th}}}$ in this section. Note that the exponent $2R_{\textrm{th}}$ is due to the two time slots required for the relaying. The SOP of the relay-aided system without the direct link can be obtained easily following the similar approach described in this section. 



\section{Numerical Results}\label{sec_result}

In this section, we present the analytical SOP using Theorem 1, the numerical integration (NI) of the SOP expressions in the integral form corresponding to the SS scheme and the OS scheme in \eqref{eq_SS_SOP2} and (\ref{eq_SOP_optimal}), respectively,  and the simulated results to validate the correctness of our derived approximate closed-form solutions.  The color red with a marker `$\times$' indicates the simulation results. Dashed horizontal lines represent the SOP in the high-SNR regime. 
It is to be noted that all the high-SNR horizontal lines are not in the range of the SOP and SNRs shown.
For numerical results, we assume $R_{\textrm{th}}=1$ bit per channel use (bpcu), and $N_0=-110$ dB \cite{nemanja_TWC}.
The other parameters for a particular figure are shown in the corresponding caption. In the figures, we denote the single-user single-eavesdropper system and single-user multiple-eavesdropper system as SE and  ME, respectively, which are the special cases of the SS and OS schemes for the proposed multi-user multiple-eavesdropper system. 
Relay-aided systems with and without a direct link from the source to the users are denoted as R-DL and R-NDL, respectively.  

\begin{figure}[ht]
\centering
\includegraphics[width=0.485\textwidth]{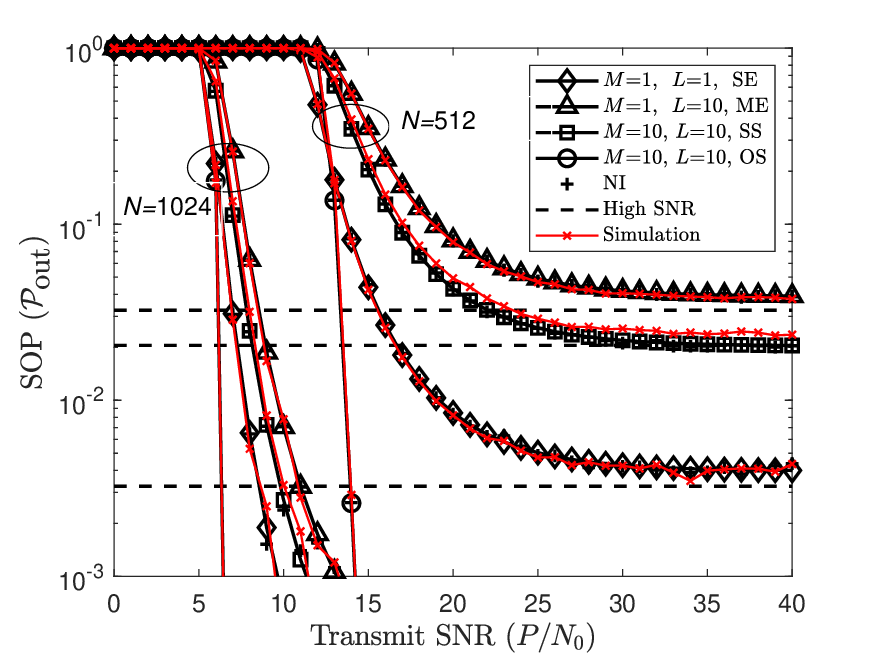}
\caption{{SOP vs. $P/N_0$ by varying $N=\{512, 1024\}$  when  $f=2$ GHz, $\{(M, L)\}=\{(1,1), (1, 10), (10,10)\}$,   $\delta_{\text{SU}}=900$ m,  $\delta_{\text{SE}}=300$ m, and $\delta_{\text{SR}} = \delta_{\text{RS}} = \delta_{\text{RU}} = \delta_{\text{RE}} = 200$ m.}}
\label{FIG2_vary_N_64_128}
\end{figure}

\subsection{Comparison of SE, ME, SS, and OS schemes}
Fig. \ref{FIG2_vary_N_64_128} plots the SOP versus the transmit SNR $P/N_0$ for the SE, ME, SS, and OS schemes for different numbers of RIS elements. The results for the SE, ME, and SS schemes are obtained using (\ref{eq_sop_suboptimal_final}), and the OS scheme is obtained using (\ref{eq_SOP_optimal}).  Though the analytical results are approximate results, these match well with the simulation due to large $N$, since CLT improves performance as the number of RVs increases.
This validates our analytical methodology. 
%
{
\begin{remark}
   We observe that as the number of RIS elements increases, the performance gap between the SS and OS schemes decreases in the SNR range where the SOP is not saturated. For example, the SNR gain attained by the OS scheme over the SS scheme is nearly 8 dB for $N=512$ and 2 dB for $N=1024$ at an SOP of $0.03$
   when  $M=10$ and $L=10$. This shows that if $N$ is large,
employing the SS scheme does not result in significant performance loss compared to
the OS scheme. 
\end{remark}
}
This is the advantage of the SS scheme as it does not require the CSI of the eavesdroppers; still, it closely follows the performance of the OS scheme if $N$ is large.



We notice that as $P/N_0$ increases, the SOP saturates to a constant level. This agrees with the finding of the SOP analysis in the high-SNR regime. The horizontal lines match well with the approximate analytical solutions at high $P/N_0$, confirming the correctness of our SOP analysis in the high-SNR regime.  The benefit of the scheduling schemes is that they can significantly reduce the saturation level.
{
\begin{remark}
The OS scheme benefits the most from an increase in $M$ and $N$ as these parameters yield a significant performance improvement compared to the SS scheme. For example, as $M$ increases from 1 to 10, the performance gap between the OS scheme and the ME scheme is greater than the performance gap between the SS scheme and the ME scheme in the saturation region.    
\end{remark}
}  This finding is consistent with the  SOP analysis in the high-SNR regime, where it is found that the performance in the OS and SS schemes is proportional to $\exp(-MN)$ and $\exp(-N)$, respectively.

\subsection{Comparison of opportunistic user scheduling with NOMA-based scheduling}
In this section, we compare the RIS-aided user scheduling schemes with the RIS-aided NOMA-based scheduling scheme. 
In the NOMA-based scheduling, a pair of users is selected independently to schedule in each time slot.
One of the users is the best user among all users, selected in a similar manner as in the SS scheme, whereas the other one is the worst user, given that the phases of the RIS elements are aligned in favor of the best user \cite{Ding_two_user_NOMA, ding_cooperative_noma_letter,Nallanathan_NOMA_RIS}.
We note that the simultaneous phase alignment of the RIS elements towards each user is difficult when a NOMA  scheme is employed \cite{NOMA_secrecy_RIS_GC_21,noma_secrecy_ISR_MIMO}, which may compromise the secrecy performance of the NOMA system. The power allocation between the best and the worst user is implemented by maximizing the sum rate of the selected pair.
In the NOMA system, we consider the \textit{sum secrecy rate outage probablity} (also shortened to SOP), defined as the probability that the sum of the secrecy rates of the two selected NOMA users falls below $R_{\textrm{th}}$. This ensures a fair comparison as both systems aim to deliver the same target sum secrecy rate. As all users are assumed to have identical channel characteristics, both scheduling schemes achieve user fairness because each user will be served, on average, with the same overall secrecy rate.

\begin{figure}[ht]
\centering
\includegraphics[width=0.485\textwidth]{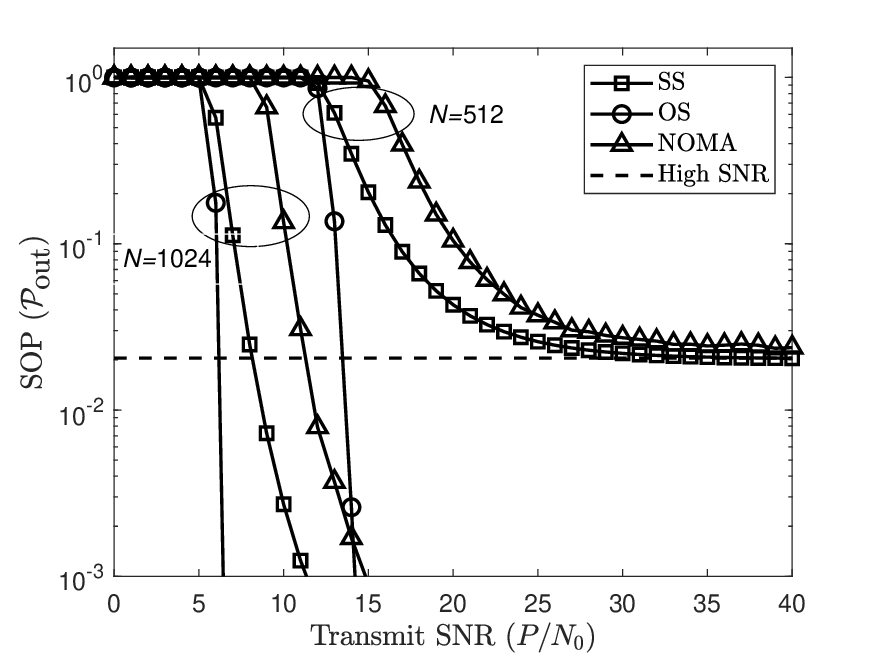}
\caption{Comparison of opportunistic scheduling and NOMA-based scheduling by varying $N=\{512, 1024\}$  when   $f=2$ GHz, $\{(M, L)\}=\{ (10,10)\}$,   $\delta_{\text{SU}}=900$ m,  $\delta_{\text{SE}}=300$ m, and $\delta_{\text{SR}} = \delta_{\text{RS}} = \delta_{\text{RU}} = \delta_{\text{RE}} = 200$ m.}
\label{FIG_vary_N_64_128_NOMA}
\end{figure} 
In Fig. \ref{FIG_vary_N_64_128_NOMA}, the SOP comparison for the SS, OS, and NOMA-based scheduling schemes is shown. 
We observe that the OS scheme always performs the best. 
We notice that the NOMA-based scheduling scheme performs worse than both the SS and the OS schemes. 
The performance of the NOMA-based scheduling scheme gradually approaches the SS scheme in the high SNR regime.
This shows that the opportunistic scheduling schemes considered in this work outperform the specific NOMA-based scheduling scheme in which the phases of the RIS elements are aligned for the best user and as a consequence, misaligned in the case of the worst user.

\subsection{Comparison of RIS-aided and DF relay-aided systems}

\begin{figure}
\centering
\includegraphics[width=0.485\textwidth]{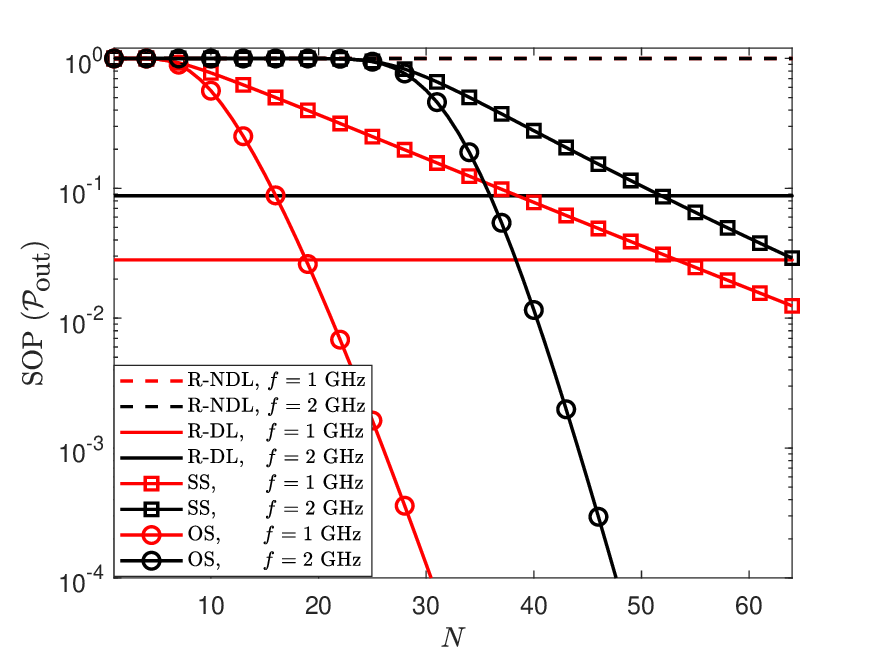}
\caption{{SOP vs. $N$ for the RIS and DF relay-aided systems by varying $f=\{1,2\}$ GHz  when $P/N_0=20$ dB, $M=10$, $L=3$,   $\delta_{\text{SR}}=40$ m, $\delta_{\text{SU}}=200$ m, $\delta_{\text{SE}}=125$ m, $\delta_{\text{RS}} = \delta_{\text{RU}} = \delta_{\text{RE}} = 30$ m,   and  $\upsilon = 3$.}}
\label{FIG_vary_PLratio_N_M10}
\end{figure} 
Fig. \ref{FIG_vary_PLratio_N_M10} shows the SOP performance comparison of the RIS-aided (SS and OS) and the DF relay-aided (R-NDL and R-DL) systems with respect to $N$ for different frequencies $f=\{1, 2\}$ GHz. The horizontal lines in the figure are the SOP of the relay-aided systems. The SOP as a function of $N$ provides a design guideline of how to choose between the RIS-aided system and the relay-aided system. The figure provides a way to find an appropriate number of RIS elements to outperform the relay-aided system for specific system parameters.
{
\begin{remark}
The RIS-aided system outperforms the relay-aided system if $N$ is sufficiently large. For instance, the RIS-aided system with SS scheme requires approximately 53 elements to outperform the R-DL system when  $f=1$GHz; however, it requires only a few elements to beat the corresponding R-NDL system.   
\end{remark}
} We also observe that as frequency increases, the performance of both the RIS and relay-aided systems degrade. 
Although higher frequency leads to increased path loss in both RIS and relay-aided systems,  the effect of increased path loss can be compensated in the RIS-aided system by increasing the number of elements.

\begin{figure}
\centering
\includegraphics[width=0.485\textwidth]{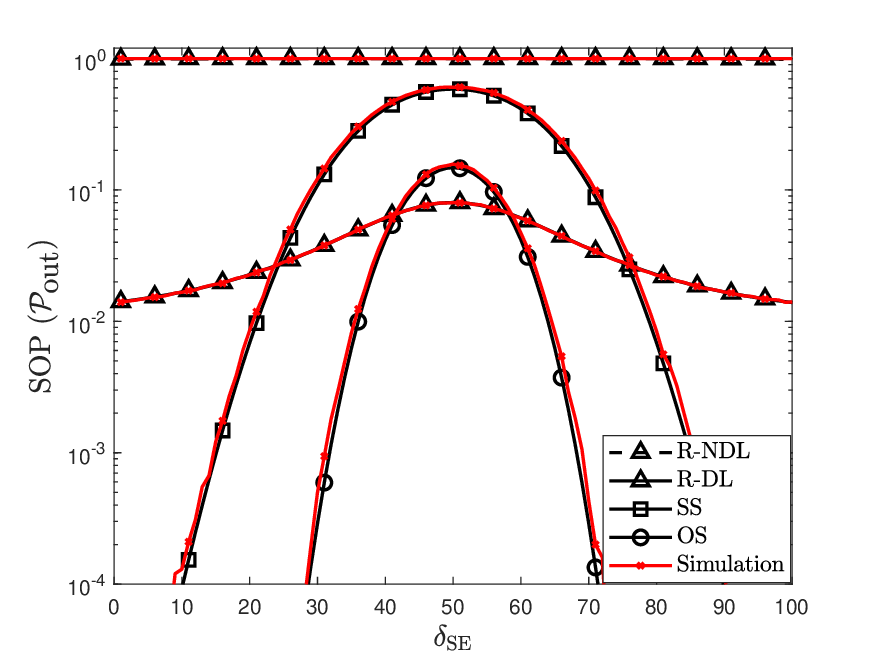}
\caption{{SOP vs. $\delta_{\text{SE}}$ for the RIS and DF relay-aided systems when $f=2$ GHz,  $N =64$,  $M=10$, $L=3$,   $\delta_{\text{SR}}=50$ m, $\delta_{\text{RS}}=\delta_{\text{RU}}=70$ m, $\delta_{\text{RE}}=20$ m,   and $\upsilon=3$.}}
\label{FIG4b_SOP_v_Delta_SE}
\end{figure}
Fig. \ref{FIG4b_SOP_v_Delta_SE}  provides an SOP performance comparison between the RIS-aided (ME, SS, and OS) and the DF relay-aided (R-NDL and R-DL) systems by varying the distance $\delta_{\text{SE}}$ between the planes containing S and $\text{E}^{(l)}$.
It can be observed from   Fig. \ref{FIG4b_SOP_v_Delta_SE} that the SOP of the RIS-aided system increases with increasing $\delta_{\text{SE}}$ until  $\delta_{\text{SE}} =\delta_{\text{SR}}$, thereafter.
This is due to the position of the RIS in the system; the secrecy performance degrades as the eavesdroppers move closer to the RIS. It is also observed that the relay-aided  (R-DL and R-NDL) systems follow a similar trend as those of the RIS-aided  (SS and OS) systems. While comparing the RIS and relay-aided systems, we notice that among the DF relay-aided systems, only  R-DL is comparable with the RIS-aided system, as the R-NDL is always outperformed. The  RIS-aided system always outperforms the corresponding R-NDL system. 
\color{black}

\section{Conclusion} 
We derive the SOP for the opportunistic user scheduling schemes (suboptimal and optimal) that incorporate multiple users and multiple eavesdroppers {in a single antenna} RIS-aided system. 
{The SOP of the best antenna-user pair scheduling scheme in the multiple antenna case can be directly achieved from the SOP of the suboptimal scheduling scheme in the single antenna system. } 
A realistic path loss model is considered in the analysis by incorporating frequency, distances, and angles of incidence and reflection at the RIS.
The closed-form approximate SOP and its simplified high-SNR expression are derived for each scheduling scheme.
Though the performance of the SS scheme is worse than that of the OS scheme, the SS scheme does not require eavesdropping CSI; however, its performance is close to the OS scheme when the number of RIS elements is large.
In the high-SNR regime, the  SOP saturates to a constant level depending on the ratio of the path loss of the source-to-users and source-to-eavesdroppers indirect links.
The high-SNR expressions can be used to decide optimal RIS placement. These expressions confirm that the SOP decreases exponentially with the number of elements and with the product of the number of elements and users in the SS scheme and in the OS scheme, respectively.
It also shows that the SOP increases linearly with the number of eavesdroppers in the SS scheme and increases with the number of eavesdroppers raised to power the number of users in the OS scheme.
A performance comparison of opportunistic schemes with a specific NOMA-based scheme reveals that the opportunistic schemes outperform the NOMA-based scheme.
Finally, it is concluded that the RIS-aided system outperforms the DF relay-aided system only if the RIS has a sufficient number of elements. This critical number depends on the frequency of operation. 

\appendix
\subsection{The solution of $\mathcal{J}_{+}^{(\boldsymbol{k},l)}((\sigma^{(\boldsymbol{k})}_{\text{U}})^2)$ given by (\ref{eq_I1_prime_subopt_ext})}\label{appendix_Jplus_SS}
We manipulate (\ref{eq_I1_prime_subopt_ext})  by absorbing the two exponential functions in it into one and 
completing the square on the argument to obtain
\begin{align}
\label{eq_J1_lemma3_addition}
\mathcal{J}_{+}^{(\boldsymbol{k},l)}((\sigma^{(\boldsymbol{k})}_{\text{U}})^2)&=\frac{1}{2\lambda_{\text{E}}^{(l)}}\exp\Bigg(\frac{\rho-1}{\rho\lambda_{\text{E}}^{(l)}}-\frac{\mu^2_{\text{U}}\bar{\Gamma}_{\text{U}}}{2(\sigma^{(\boldsymbol{k})}_{\text{U}})^2\rho\lambda_{\text{E}}^{(l)}\Upsilon^{(\boldsymbol{k},l)}}\Bigg) \nn \\
&\times\int_{0}^{\infty}
\exp\Bigg(-\Upsilon^{(\boldsymbol{k},l)}\Bigg(\sqrt{\frac{\rho-1+\rho x}{\bar{\Gamma}_{\text{U}}}}\nn\\
&-\frac{\mu_{\textrm{U}}}{2(\sigma^{(\boldsymbol{k})}_{\text{U}})^2\Upsilon^{(\boldsymbol{k},l)}}\Bigg)^2\Bigg)dx,
\end{align}
where
$\Upsilon^{(\boldsymbol{k},l)}= \frac{1}{2(\sigma^{(\boldsymbol{k})}_{\text{U}})^2}+\frac{\bar{\Gamma}_{\text{U}}}{\rho\lambda_{\text{E}}^{(l)}}$. 
Next, after making the change of variable $t=\sqrt{\frac{\rho-1+\rho x}{\bar{\Gamma}_{\text{U}}}}-\frac{\mu_{\textrm{U}}}{2(\sigma^{(\boldsymbol{k})}_{\text{U}})^2
\Upsilon^{(\boldsymbol{k},l)}}$ and performing some algebraic manipulations, we obtain
\begin{align}
\label{eq_def_J1_correction}
\mathcal{J}_{+}^{(\boldsymbol{k},l)}((\sigma^{(\boldsymbol{k})}_{\text{U}})^2)
&=\frac{ \bar{\Gamma}_{\text{U}}}{\rho\lambda_{\text{E}}^{(l)}}
\exp\Bigg(\frac{ \rho-1}{\rho\lambda_{\text{E}}^{(l)}}-\frac{ \mu^2_{\text{U}}\bar{\Gamma}_{\text{U}}}
{2(\sigma^{(\boldsymbol{k})}_{\text{U}})^2\rho\lambda_{\text{E}}^{(l)}\Upsilon^{(\boldsymbol{k},l)}}\Bigg)\nn\\
&\times\int_{\sqrt{\frac{\rho-1}{\bar{\Gamma}_{\text{U}}}}-\frac{\mu_{\textrm{U}}}{2(\sigma^{(\boldsymbol{k})}_{\text{U}})^2\Upsilon^{(\boldsymbol{k},l)}}}
^{\infty}
\Big(t+\frac{\mu_{\textrm{U}}}{2(\sigma^{(\boldsymbol{k})}_{\text{U}})^2\Upsilon^{(\boldsymbol{k},l)}}\Big)\nn\\
&\times\exp\lb(-\Upsilon^{(\boldsymbol{k},l)}t^2\rb)
 dt.
\end{align}
The solution of the integral in (\ref{eq_def_J1_correction}) is obtained from \cite[eq.
(3.321)]{book_Gradshteyn_Ryzhik} 
and subsequently, the result is written in (\ref{eq_ss_j_new_theorem}).

\subsection{The solution of $\mathcal{I}_{+}^{(\boldsymbol{k},l)}((\sigma^{(\boldsymbol{k})}_{\text{U}})^2)$ given by   \eqref{eq_SS_I1_appendix2}}
\label{appendix_Iplus_SS}

We manipulate \eqref{eq_SS_I1_appendix2} by first absorbing the two exponential functions into one and then completing the square on the argument of the exponential function following the steps taken in Appendix \ref{appendix_Jplus_SS} to achieve
\begin{align}
\label{eq_I1_lemma1_2}
\mathcal{I}_{+}^{(\boldsymbol{k},l)}((\sigma^{(\boldsymbol{k})}_{\text{U}})^2)
&=\frac{1}{2\lambda_{\text{E}}^{(l)}}\exp\Bigg(\frac{\rho-1}{\rho\lambda_{\text{E}}^{(l)}}-\frac{\mu^2_{\text{U}}\bar{\Gamma}_{\text{U}}}{2(\sigma^{(\boldsymbol{k})}_{\text{U}})^2\rho\lambda_{\text{E}}^{(l)}\Upsilon^{(\boldsymbol{k},l)}}\Bigg) \nn \\
&\times\int_{\frac{\mu^2_{\text{U}}\bar{\Gamma}_{\text{U}}-\lb(\rho-1\rb)}{\rho}}^{\infty}
\exp\Bigg(-\Upsilon^{(\boldsymbol{k},l)}\Bigg(\sqrt{\frac{\rho-1+\rho x}{\bar{\Gamma}_{\text{U}}}}\nn\\
&-\frac{\mu_{\textrm{U}}}{2(\sigma^{(\boldsymbol{k})}_{\text{U}})^2\Upsilon^{(\boldsymbol{k},l)}}\Bigg)^2\Bigg)dx.
\end{align}
On performing a change of variables $t=\sqrt{\frac{\rho-1+\rho x}{\bar{\Gamma}_{\text{U}}}}-\frac{\mu_{\textrm{U}}}{2(\sigma^{(\boldsymbol{k})}_{\text{U}})^2
\Upsilon^{(\boldsymbol{k},l)}}$ and after some further algebraic manipulations, we obtain
\begin{align}
\label{eq_def_I1_3}
&\mathcal{I}_{+}^{(\boldsymbol{k},l)}((\sigma^{(\boldsymbol{k})}_{\text{U}})^2)=\frac{ \bar{\Gamma}_{\text{U}}}{\rho\lambda_{\text{E}}^{(l)}}
\exp\Bigg(\frac{ \rho-1}{\rho\lambda_{\text{E}}^{(l)}}-\frac{ \mu^2_{\text{U}}\bar{\Gamma}_{\text{U}}}
{2(\sigma^{(\boldsymbol{k})}_{\text{U}})^2\rho\lambda_{\text{E}}^{(l)}\Upsilon^{(\boldsymbol{k},l)}}\Bigg)\nn\\&
\times\int_{\frac{\mu_{\textrm{U}}\bar{\Gamma}_{\text{U}}}{\rho\lambda_{\text{E}}^{(l)}\Upsilon^{(\boldsymbol{k},l)}}}
^{\infty}
\Big(t+\frac{\mu_{\textrm{U}}}{2(\sigma^{(\boldsymbol{k})}_{\text{U}})^2\Upsilon^{(\boldsymbol{k},l)}}\Big)
\exp\lb(-\Upsilon^{(\boldsymbol{k},l)} t^2\rb)
 dt.
\end{align}
The closed-form solution of (\ref{eq_def_I1_3}) is obtained with the help of \cite[eq.
(3.321)]{book_Gradshteyn_Ryzhik}
and finally, its result is written in \eqref{eq_SS_I_new}.

\subsection{The solution of  $\mathcal{J}_{+}^{(i,l)}(({\sigma^{(i)}_{\text{U}}})^2)$ given by (\ref{eq_Jplus_approx})}
\label{appendix_proof_J_plus}
The solution of  $\mathcal{J}_{+}^{(i,l)}(({\sigma^{(i)}_{\text{U}}})^2)$ is obtained following Appendix \ref{appendix_Jplus_SS}.  By first absorbing the two exponential functions in (\ref{eq_Jplus_approx}) into one, then completing the square on the argument of the exponential function, and finally applying a change of variables assuming $t=\sqrt{\frac{\rho x}{\bar{\Gamma}_{\text{U}}}}-\frac{\mu_{\textrm{U}}}{2({\sigma^{(i)}_{\text{U}}})^2\Upsilon^{(i,l)}}$,  where
$\Upsilon^{(i,l)}=\frac{1}{2({\sigma^{(i)}_{\text{U}}})^2}+\frac{ \bar{\Gamma}_{\text{U}}}{\rho\lambda_{\text{E}}^{(l)}}$, we achieve \begin{align}
\label{eq_I2_approx_lemma7}
&\mathcal{J}_{+}^{(i,l)}(({\sigma^{(i)}_{\text{U}}})^2)=\frac{ \bar{\Gamma}_{\text{U}}}{\rho\lambda_{\text{E}}^{(l)}}
\exp\Bigg(\frac{ -\mu^2_{\text{U}}\bar{\Gamma}_{\text{U}}}
{2({\sigma^{(i)}_{\text{U}}})^2\rho\lambda_{\text{E}}^{(l)}\Upsilon^{(i,l)}}\Bigg) \nn \\
&\times\int_{-\frac{\mu_{\textrm{U}}}{2({\sigma^{(i)}_{\text{U}}})^2\Upsilon^{(i,l)}}}
^{\infty}
\Big(t+\frac{\mu_{\textrm{U}}}{2({\sigma^{(i)}_{\text{U}}})^2\Upsilon^{(i,l)}}\Big)
\exp\lb(-\Upsilon^{(i,l)}t^2\rb)
 dt.
\end{align}
The closed-form solution of the above integral is obtained with the help of \cite[eq.(3.321)]{book_Gradshteyn_Ryzhik}. 
The result is provided in (\ref{eq_Jplus_closed}).

\subsection{The solution of  $\mathcal{I}_{+}^{(\boldsymbol{k},l)}((\sigma^{(\boldsymbol{k})}_{\text{U}})^2)$ given by (\ref{eq_Iplus_high_SNR}) }
\label{appendix_proof_I_plus}
The solution of $\mathcal{I}_{+}^{(\boldsymbol{k},l)}((\sigma^{(\boldsymbol{k})}_{\text{U}})^2)$ is
 obtained following the steps taken in Appendix \ref{appendix_proof_J_plus}. The first two steps implemented in (\ref{eq_Iplus_high_SNR}) are the same as in Appendix \ref{appendix_proof_J_plus}, then by applying a change of variables assuming $t=\sqrt{\frac{\rho x}{\bar{\Gamma}_{\text{U}}}}-\frac{\mu_{\textrm{U}}}{2({\sigma^{(\boldsymbol{k})}_{\text{U}}})^2
\Upsilon^{(\boldsymbol{k},l)}}$ and doing some further algebraic manipulations in (\ref{eq_Iplus_high_SNR}) we obtain
\begin{align}
\label{eq_SS_I_new_asym_lemma8}
&\mathcal{I}_{+}^{(i,l)}(({\sigma^{(\boldsymbol{k})}_{\text{U}}})^2)=\frac{\bar{\Gamma}_{\text{U}}}{\rho\lambda_{\text{E}}^{(l)}}
\exp\Bigg(\frac{ -\mu^2_{\text{U}}\bar{\Gamma}_{\text{U}}}
{2({\sigma^{(\boldsymbol{k})}_{\text{U}}})^2\rho\lambda_{\text{E}}^{(l)}\Upsilon^{(\boldsymbol{k},l)}}\Bigg) \nn \\
&\times\int_{\frac{\mu_{\textrm{U}}\bar{\Gamma}_{\text{U}}}{\rho\lambda_{\text{E}}^{(l)}\Upsilon^{(\boldsymbol{k},l)}}}
^{\infty}
\Big(t+\frac{\mu_{\textrm{U}}}{2({\sigma^{(\boldsymbol{k})}_{\text{U}}})^2\Upsilon^{(\boldsymbol{k},l)}}\Big)
\exp\lb(-\Upsilon^{(\boldsymbol{k},l)} t^2\rb)
 dt.
\end{align}
The solution of (\ref{eq_SS_I_new_asym_lemma8}) in closed form is obtained with the help of \cite[eq.
(3.321.1)]{book_Gradshteyn_Ryzhik} and is provided in  (\ref{eq_Iplus_high_SNR_closed}). 

\bibliographystyle{IEEEtran}
\bibliography{IEEEabrv,RIS}
\end{document}